\documentclass[11pt]{article}

\usepackage{libertine,inconsolata,amsfonts,amssymb}
\usepackage{float}
\usepackage{geometry}
\usepackage[hyphens]{url}
\usepackage[colorlinks=true,allcolors=blue]{hyperref}

\usepackage{amsmath,mathtools,physics}
\let\cclass\textrm
\let\P\relax
\newcommand{\P}{\cclass{P}}
\newcommand{\BQP}{\cclass{BQP}}
\newcommand{\BQNC}{\cclass{BQNC}}
\newcommand{\PSPACE}{\cclass{PSPACE}}
\newcommand{\BPP}{\cclass{BPP}}
\newcommand{\HQC}{\cclass{HQC}}
\newcommand{\FewTierHQC}{\cclass{FewTierHQC}}
\newcommand{\JC}{\cclass{JC}}
\newcommand{\NC}{\cclass{NC}}

\let\exp\relax
\newcommand{\exp}{\cclass{exp}}
\newcommand{\polylog}{\cclass{polylog}}
\newcommand{\poly}{\cclass{poly}}
\newcommand{\subexp}{\cclass{subexp}}

\usepackage[linesnumbered,ruled,vlined,algosection]{algorithm2e}
\SetCommentSty{textsf}

\usepackage{tikz}
\usepackage{tkz-graph}
\usetikzlibrary{positioning,patterns,calc,shapes}

\usepackage{amsthm}
\newtheorem{theorem}{Theorem}[section]

\newtheorem{lemma}[theorem]{Lemma}

\newtheorem{conjecture}[theorem]{Conjecture}
\newtheorem{problem}[theorem]{Problem}
\theoremstyle{definition}
\newtheorem{definition}[theorem]{Definition}

\theoremstyle{remark}
\newtheorem{remark}[theorem]{Remark}

\newcommand{\NN}{\ensuremath{\mathbb{N}}}

\newcommand{\fid}{\operatorname{F}}

\newcommand{\bigabs}[1]{\bigl\lvert #1 \bigr\rvert}
\newcommand{\Bigabs}[1]{\Bigl\lvert #1 \Bigr\rvert}

\newcommand{\out}{\mathtt{OUT}}
\newcommand{\vout}{\mathtt{VKNOWN}}
\newcommand{\vhout}{\mathtt{VHKNOWN}}

\newcommand{\tnorm}[1]{\norm{#1}_1}
\newcommand{\cdist}[2]{\tnorm{#1 - #2}}

\newcommand{\entrance}{\textsc{entrance}}
\newcommand{\exit}{\textsc{exit}}

\newcommand{\invalid}{\texttt{INVALID}}

\newcommand{\Vknown}{V_{\text{known}}}
\newcommand{\Viknown}{V_{\text{knowninit}}}
\newcommand{\Vfknown}{V_{\text{knownfinal}}}
\newcommand{\Vhknown}{V_{\text{known}}^{\textit{hist}}}
\newcommand{\Vcknown}{V_{\text{known}}^{\textit{current}}}
\newcommand{\Vnknown}{V_{\text{known}}^{\textit{new}}}
\newcommand{\Outliers}{\text{Outliers}}
\newcommand{\abort}{\texttt{ABORT}}
\newcommand{\emptydict}{\{\}}
\newcommand{\while}{\texttt{while}}

\newcommand{\A}{\mathcal{A}}
\let\C\relax
\newcommand{\C}{\mathcal{C}}

\newcommand{\J}{\mathcal{J}}
\renewcommand{\H}{\mathcal{H}}
\newcommand{\M}{\mathcal{M}}
\newcommand{\T}{\mathcal{T}}

\newcommand{\QuantumTierSimulator}{\textrm{QuantumTierSimulator}}
\newcommand{\ClassicalTierSimulator}{\textrm{ClassicalTierSimulator}}
\newcommand{\ClassicalSimulationWrapper}{\textrm{ClassicalSimulationWrapper}}
\newcommand{\JozsaQuantumTierSimulator}{\textrm{JozsaQuantumTierSimulator}}
\newcommand{\BottleneckQuantumTierSimulator}{\textrm{BottleneckQuantumTierSimulator}}
\newcommand{\SimulateOracle}{\textrm{SimulateOracle}}
\newcommand{\bottleneck}{\textrm{Bottleneck}}
\newcommand{\merge}{\textrm{Merge}}
\newcommand{\QuantumLayerSimulator}{\textrm{QuantumLayerSimulator}}

\DeclareMathOperator*{\dsPr}{\mathbb{P}}
\let\Pr\relax
\DeclareMathOperator{\Pr}{\mathbb{P}}

\title{\bfseries\Large Computations with Greater Quantum Depth Are Strictly More Powerful (Relative to an Oracle)}

\author{Matthew Coudron\thanks{%
        NIST/QuICS, University of Maryland, USA.
        Email: \texttt{mcoudron@umd.edu}.
        This work was completed while at the IQC, University of Waterloo, Canada.
    }
    \and
    Sanketh Menda\thanks{%
        University of Waterloo, Canada.
        Email: \texttt{sgmenda@uwaterloo.ca}.
    }
}

\date{April 22, 2020}

\begin{document}
\sloppy

\maketitle
\begin{abstract}
    A conjecture of Jozsa (arXiv:quant-ph/0508124) states that any
    polynomial-time quantum computation can be simulated by
    polylogarithmic-depth quantum computation interleaved with polynomial-depth
    classical computation.  Separately, Aaronson conjectured that there exists
    an oracle $\mathcal{O}$ such that $\textrm{BQP}^{\mathcal{O}} \neq
    (\textrm{BPP}^\textrm{BQNC})^{\mathcal{O}}$.  These conjectures are
    intriguing allusions to the unresolved potential of combining classical and
    low-depth quantum computation.  In this work we show that the Welded Tree
    Problem, which is an oracle problem that can be solved in quantum polynomial
    time as shown by Childs et al. (arXiv:quant-ph/0209131), cannot be solved in
    $\textrm{BPP}^{\textrm{BQNC}}$, nor can it be solved in the class that Jozsa
    describes.  This proves Aaronson's oracle separation conjecture and provides
    a counterpoint to Jozsa's conjecture relative to the Welded Tree oracle
    problem.  More precisely, we define two complexity classes, $\textrm{HQC}$
    and $\textrm{JC}$ whose languages are decided by two different families of
    interleaved quantum-classical circuits. $\textrm{HQC}$ contains
    $\textrm{BPP}^\textrm{BQNC}$ and is therefore relevant to Aaronson's
    conjecture, while $\textrm{JC}$ captures the model of computation that Jozsa
    considers. We show that the Welded Tree Problem gives an oracle separation
    between either of $\{\textrm{JC}, \textrm{HQC}\}$ and $\textrm{BQP}$.
    Therefore, even when interleaved with arbitrary polynomial-time classical
    computation, greater ``quantum depth'' leads to strictly greater
    computational ability in this relativized setting.
\end{abstract}

\section{Introduction}

\subsection{The Power of Hybrid Quantum Computation}

Our work is inspired by the following conjecture in quantum computing folklore.
Variants of this conjecture have been considered by Jozsa~\cite{Jozsa06} and
Aaronson~\cite{Aaronson05, Aaronson11, Aaronson14}.

\begin{conjecture}[Folklore]
    \label{conj:folklore}
    Any polynomial-time quantum computation can be simulated by a polynomial-size
    classical computation interleaved with polylogarithmic-depth quantum
    computation.
\end{conjecture}

Intriguingly, this conjecture is known to hold for some of the most influential
quantum algorithms.  For example, Cleve and Watrous~\cite{CleveW00} showed that
Shor's algorithm for \textsc{Factoring} can be implemented using log-depth
polynomial-size quantum circuits with polynomial-time classical pre- and
post-processing.  Indeed, one might be able to use a similar methodology to
parallelize many quantum algorithms that rely on the quantum Fourier transform.

Similarly, most oracle separations that show a quantum speedup seem to be
consistent with Conjecture~\ref{conj:folklore}. The prototypical problems that
exhibit an exponential quantum speedup---\textsc{Simon's problem}~\cite{Simon97}
and \textsc{Forrelation}~\cite{Aaronson10, AaronsonA18}---can both be solved using
\emph{constant-depth} quantum circuits with oracle access, and polynomial-time
classical pre- and post-processing.  Thus, while they constitute oracle
separations between $\P$ and $\BQP$, they do not, on their own, suggest an
oracle separation between $\BQP$ and the sort of class that
Conjecture~\ref{conj:folklore} considers.  All of this could be taken as an
indication that the class described in Conjecture~\ref{conj:folklore} is very
powerful.

There are at least two seemingly incomparable interpretations of the model of
computation considered in Conjecture~\ref{conj:folklore}.  One was proposed by
Jozsa~\cite{Jozsa06}, and could be described in shorthand as version of
$\BQNC^\BPP$.  Another, $\BPP^\BQNC$, was considered by
Aaronson~\cite{Aaronson05, Aaronson11, Aaronson14}, who conjectured an oracle
separation between this class and $\BQP$.  In our work we will define the first
class under the name $\JC$ and the second class under the name $\HQC$ in order
to avoid confusion about the oracle access model for these classes, and because
of other technicalities.  These complexity classes may be considered
hardware-motivated mathematical models in the sense that they endeavor to
capture the computational problems which can be solved by a quantum device of
limited depth (due either to limited coherence time, or some other restriction)
when combined, in one of two reasonable ways, with classical side-processing of
arbitrary polynomial depth.

Note that, in order to disprove either interpretation of
Conjecture~\ref{conj:folklore}, one would necessarily need to separate $\P$ from
$\BQP$ as a prerequisite, and such a statement may be very difficult to prove as
an unconditional mathematical fact.  For example, it would require separating
$\P$ from $\PSPACE$.  In this work we will prove a separation between $\BQP$ and
both of $\{\JC, \HQC\}$ relative to the Welded Tree oracle.  At the end, we also
remark on the possibility of extending this result to a separation based on a
cryptographic assumption.

\subsection{Results and Organization}

In this subsection we will give more detailed (yet still informal) statements of
our results, and provide pointers to the sections of this paper which state and
prove each result formally.  Section \ref{sec:preliminaries} of this paper
covers preliminaries which are essential to formalizing our results, including
the definitions of complexity classes $\JC$ and $\HQC$ referred to in the
abstract, as well as a discussion of the oracle-access model for these classes,
and background on the Welded Tree Problem.

\subsubsection{Hybrid Quantum Computation with Few Tiers}
\label{sec:hybr-quant-comp}

Section \ref{sec:few-tiers} of this paper is a warm-up meant to build intuition
for the techniques used in our main results.  In Section \ref{sec:few-tiers} we
show how our techniques can easily prove that a limited class of Hybrid Quantum
Circuits, those with ``Few Tiers'', cannot be used to solve the Welded Tree
Problem with high probability.  The result can be summed up in the following
theorem.

\begin{theorem}[Informal]\label{thm:few-tier-informal}
    No quantum algorithm with oracle access to the Welded Tree oracle, and using
    only $O(\polylog(n))$-depth quantum circuits, alternated with polynomial time
    classical computations at most $\sqrt{n}$ times, can solve the Welded Tree
    Oracle problem with probability higher than
    $2^{-\Omega(n)}$.
\end{theorem}

\paragraph{Proof Overview.} Our analysis for
Theorem~\ref{thm:few-tier-informal}, and for the main theorems in this paper as
well, works by exhibiting a classical simulation algorithm for the hybrid
quantum oracle algorithm in question.  Our classical simulation produces an
exponentially close approximation of the output of the original hybrid quantum
algorithm for this oracle problem, and yet it uses only sub-exponentially many
classical oracle queries to do so.  It then follows by the known classical lower
bound for the Welded Tree Problem, due to Childs et al.~\cite{ChildsCDFGS03}
(see Theorem \ref{thm:class-lower-bound}), that the hybrid quantum algorithm
cannot possibly be solving the Welded Tree Problem with better than
(sub-)exponentially small probability.

Our classical simulation algorithm makes use of the fact that, in the Welded
Tree Problem, the $2n$-bit labels of valid vertices of the welded tree are
randomly chosen from among the exponentially larger set of all $2n$-bit strings.
See Subsection \ref{subsec:weldedtree} for more details on the definition the
Welded Tree Problem.  Therefore, while a quantum algorithm can query the welded
tree oracle in superposition for information about every $2n$-bit string, it is
intuitive that all but an exponentially small fraction of the mass of that
quantum query will be supported on either invalid $2n$-bit strings, or on valid
$2n$-bit strings which correspond to vertices that are adjacent to previously
known vertices in the welded tree graph.  If this were not so it would imply a
classical query algorithm which can non-trivially guess a valid-but-unknown
vertex label using sub-exponentially many queries, and that contradicts known
results of \cite{ChildsCDFGS03}.  Therefore, a classical query algorithm can
produce a classical description of a quantum state which is exponentially close
to the true output of a quantum oracle query in exponential time and using
sub-exponentially many classical oracle queries.  This is done simply by
assuming that the oracle returns $\invalid$ on any input in the superposition
which is not a previously encountered valid vertex label.  Moreover, since the
number of valid vertex labels in the first $\sqrt{n}$ levels of the welded tree
is sub-exponential, this classical simulation can be performed for
$\sqrt{n}$-depth relativized quantum circuits using only sub-exponentially many
classical queries.  See Section \ref{sec:few-tiers} for a rigorous discussion.

\subsubsection{Jozsa's Conjecture}

In Section \ref{sec:relat-jozs-conj} of this paper we observe that the
techniques of Section \ref{sec:few-tiers} can be augmented to handle a class of
circuits considered by Jozsa.

\begin{quote}
    ``\textbf{Conjecture:} Any polynomial time quantum algorithm can be implemented
    with only $O(\log n)$ quantum layers interspersed with polynomial time classical
    computations.''

    Richard Jozsa~\cite[Section 8]{Jozsa06}
\end{quote}

Our result is summarized in the following statement, which is formalized and
proven as Theorem \ref{thm:jozsa-formal} in this paper.

\begin{theorem}[Informal] \label{thm:jozsa-informal}
    No quantum algorithm with oracle access to the Welded Tree oracle, and using
    only $O(\polylog(n))$ quantum layers of polynomial width, interspersed with
    polynomial time classical computations, can solve the Welded Tree Problem
    with probability higher than $2^{-\Omega(n)}$.
\end{theorem}

Since the Welded Tree Problem can be solved in $\BQP$, as shown in
\cite{ChildsCDFGS03}, our theorem implies an oracle separation between $\BQP$
and the class that Jozsa describes.  That class will be defined precisely under
the name $\JC$ in Section \ref{sec:preliminaries}.

\paragraph{Proof Overview.} The principal difference between a Jozsa circuit and
a hybrid quantum circuit is the way interleaving works. A Jozsa circuit is a
polylogarithmic depth quantum circuit with polynomial size classical circuits
embedded in it. Moreover, the input to the embedded classical circuits is a
classical bitstring, produced by measuring a subset of the output of the
previous quantum layer in the classical basis. (See
Section~\ref{sec:jozsas-class} for a formal definition.) Observe that the
classical circuits can be decoupled from this interleaving. Therefore we can
repurpose the analysis from Subsubsection~\ref{sec:hybr-quant-comp}.

\subsubsection{Full Hybrid Quantum Computation}

In Section \ref{sec:hqc} of this paper we extend the analysis from Section
\ref{sec:few-tiers} to handle Hybrid Quantum Circuits with any polynomial number
of tiers.  Doing so requires introducing a new concept which we refer to as an
``Information Bottleneck'', together with new techniques to formalize the use of
this concept. Our main result, which is summarized in the statement below, is
formalized and proven as Theorem \ref{thm:mainthm-formal} in this paper.

\begin{theorem}[Informal]\label{thm:mainthm-informal}
    No quantum algorithm with oracle access to the Welded Tree oracle, and using
    only $O(\polylog(n))$-depth quantum circuits, alternated with polynomial
    time classical computations, polynomially many times, can solve the Welded
    Tree Oracle problem with probability higher than $2^{-\Omega(n)}$.  
\end{theorem}

Theorem \ref{thm:mainthm-informal} implies Aaronson's conjecture that there
exists an oracle $\mathcal{O}$, namely, the Welded Tree oracle, such that
$\BQP^{\mathcal{O}} \neq (\BPP^\BQNC)^{\mathcal{O}}$ \cite{Aaronson05,
Aaronson11, Aaronson14}.

\paragraph{Proof Overview.} Let us begin by explaining why the proof of Theorem
\ref{thm:mainthm-informal} seems to require a new idea from that of Theorem
\ref{thm:few-tier-informal}.  The reason that the classical simulation Algorithm
\ref{alg:classicalsim1} in Section \ref{sec:few-tiers} is sufficient to prove
Theorem \ref{thm:adaptive-query-lower-bound} is that the set of known vertices,
$\Vknown$, constructed by Algorithm \ref{alg:classicalsim1}, grows only by a
multiple of $n$ for each quantum tier encountered.  This means that, when
simulating a circuit with only $\sqrt{n}$ quantum tiers, the set $\Vknown$ only
has subexponential size by the end of Algorithm \ref{alg:classicalsim1} and the
Algorithm is, therefore, still subject to the classical lower bound
Theorem~\ref{thm:class-lower-bound}.  However, this, while encouraging, is not
sufficient to prove our new Theorem \ref{thm:mainthm-formal}, because $\HQC$ allows
circuits in which the number of tiers is some arbitrary polynomial $\poly(n)$ in
the input size $n$.  So, the set $\Vknown$ constructed by Algorithm
\ref{alg:classicalsim1} could grow to exponential size when attempting to
simulate an arbitrary circuit in $\HQC$, and this would not meet the
prerequisites for applying Theorem \ref{thm:class-lower-bound}.

To prove Theorem \ref{thm:mainthm-formal} we employ a new idea which we will
refer to informally as the ``Information Bottleneck''.  The intuition is that,
while a circuit in $\HQC$ may have $\poly(n)$ tiers, its width is also bounded
by some polynomial $g(n)$.  Therefore, while Algorithm \ref{alg:classicalsim1}
tracks a set $\Vknown$ of known vertices that grows exponentially large as it
increases through $\poly(n)$ tiers, it seems intuitive that after the end of
each quantum tier, only $g(n)$ (or, say, $\poly(n,g(n))$ at the most) of those
known vertices should ``actually matter'' to the $\HQC$ circuit being simulated.
This is because the width of the circuit bounds the amount of classical
information that can be passed from one tier of the circuit to the next. Note
that the information passed between tiers is necessarily a classical bit string
by definition.  In Section \ref{sec:hqc} will make this intuition more formal
and use it to prove an oracle separation between $\HQC$ and $\BQP$.

\subsection{Concurrent Work}

Independent and concurrent to our work Chia, Chung, and Lai \cite{Chia19} also
investigated Conjecture~\ref{conj:folklore} in its multiple manifestations, and
proved that the conjecture is false relative to an oracle.  In particular their
work also proves the oracle separation conjectured by Aaronson \cite{Aaronson05,
Aaronson11, Aaronson14}, and gives an oracle separation against Jozsa's
conjecture.  To do this they use a different oracle problem, of their own
construction, and so their analysis is very different from ours. This provides a
thought-provoking alternative approach to studying computations which require
large quantum depth.

The starting point for the oracle construction of \cite{Chia19} is
\textsc{Simon's problem}~\cite{Simon97}, but since that problem can be solved
with just $O(1)$ quantum depth, the authors of \cite{Chia19} construct a
\emph{lifted} version of Simon's problem by using ideas from cryptography to
augment the problem in such a way that it requires higher quantum depth.  In
particular, they use techniques of \emph{pointer chasing} and \emph{domain
hiding} to construct a variant of Simon's problem in which the valid domain of
the candidate Simon function is hidden in an exponentially larger set of
strings.  The intuition is that only a high-depth computation could continue
pointer chasing for long enough to identify a valid element of the domain with
high probability.  Furthermore, since one requires a uniform superposition over
valid domain elements in order to implement Simon's algorithm, one intuitively
needs high \emph{quantum} depth to solve their oracle problem.  Interleaving
classical computation and low-depth quantum computation is not sufficient
because the required superposition over valid domain elements cannot be obtained
in such a model.  Considerable technical work is required to formalize this
argument, see \cite{Chia19}.

Chia, Chung, and Lai~\cite{Chia19} also go a step further, showing that their
oracle problem separates computations with quantum depth $d$ from those with
quantum depth $2d+1$.  This represents a sharpening of the separation between
hybrid quantum circuits and $\BQP$.  We have not carefully considered whether
there is a modification of the Welded Tree Problem which obtains a similar
separation.

As for further alternative interpretations of Conjecture \ref{conj:folklore}
which have not yet been considered, we believe that the techniques in both works
extend to establish an oracle separation between the natural hierarchy of hybrid
models, $\BPP^{\BQNC^{\BPP^{\cdots}}}$ and $\BQP$ \cite{Chia19b}.

\section{Definitions, Background, and Notation}\label{sec:preliminaries}

\subsection{Hybrid Quantum Computation}

In order to make our investigation of Conjecture~\ref{conj:folklore} more
precise, in this subsection we define a hierarchy of complexity classes based on
hybrid quantum computation. We begin by setting notation, followed by a
definition of the hierarchy, and ending with a formal statement of
Conjecture~\ref{conj:folklore}.

\subsubsection{Quantum Circuits}\label{sec:quantum-circuits}

For an introduction to uniform circuit families, quantum circuits, and the
complexity classes not defined here see Watrous~\cite{Watrous09}.

The results in this paper are not sensitive to a choice of (reasonable,
universal) gate set. Nevertheless, for concreteness, we assume that our
classical circuits are composed of \emph{Toffoli} gates, and that our quantum
circuits are composed of \emph{Hadamard}, \emph{Toffoli}, and \emph{Phase}
gates. In addition, these circuits may contain \emph{query gates} (discussed in
Section~\ref{sec:oracl-quant-world}), \emph{auxiliary qubit gates} (which take
no input and produce a qubit in the $\ket{0}$ state), and \emph{garbage gates}
(which take an input and produce no output.) For an introduction to quantum
circuits, see Watrous~\cite{Watrous11}.

We assume, for simplicity and without loss of generality, that our quantum
circuits have the following form: we receive an $n$-bit input, which is then
padded with $p(n)$ qubits in the $\ket{0}$ state for some fixed polynomial $p$,
we apply a unitary---the \emph{unitary purification} of this circuit---to these
$n+p(n)$ qubits, and measure the first qubit in the computational basis and
consider that to be the output. In cases where we expect an $s$-bit output, we
measure the first $s$ qubits (and if $s$ is greater than the number of output
qubits, we pad zeros to the end of the output.)

\begin{definition}
    Define a \emph{$(m,s)$-classical layer} to be an $m$-input, $s$-output,
    depth-$1$ classical circuit.
\end{definition}

\begin{definition}
    Define a \emph{$(m,s)$-quantum layer} to be an $m$-input, $s$-output,
    depth-$1$ quantum circuit.
\end{definition}

\begin{definition}
    We say that two consecutive circuits are \emph{compatible} if the number of
    outputs of the first circuit is greater than or equal to the number of inputs
    of the second circuit. It is assumed that the extra outputs of the first
    circuit are traced out.
\end{definition}

\begin{definition}
    Define a \emph{$(m,s,d)$-classical tier} to be an $m$-input, $s$-output,
    depth-$d$ classical circuit. In other words, a $(m,s,d)$-classical tier
    consists of $d$ compatible classical layers composed with each other.
\end{definition}

\begin{definition}
    Define an \emph{$(m,s,d)$-quantum tier} to be an $m$-input $s$-output
    depth-$d$ quantum circuit followed by a measurement in the computational
    basis. In other words, a $(m,s,d)$-quantum tier consists of $d$ compatible
    quantum layers composed with each other, followed by a measurement in the
    computational basis.
\end{definition}

\begin{definition} \label{def:hybridcircuit}
    Define a \emph{$(n,\eta,c,q,g)$-hybrid-quantum circuit} $H$ to be a
    composition of $\eta$ circuits
    \begin{equation}
        C_1 \circ C_2 \circ \cdots \circ C_\eta
    \end{equation}
    such that the following hold.
    \begin{enumerate}
        \item $C_1$ is an $(n, g, c)$-classical tier.
        \item $C_\eta$ has at least one output.
        \item For odd $i > 1$, $C_i$ is an $(g, g, c)$-classical tier.
        \item For even $i$, $C_i$ is an $(g, g, q)$-quantum tier.
        \item $g(n)$ is the width of $H$.
    \end{enumerate}
    We define the output of this circuit to be the first bit of the output of
    $C_\eta$.
\end{definition}

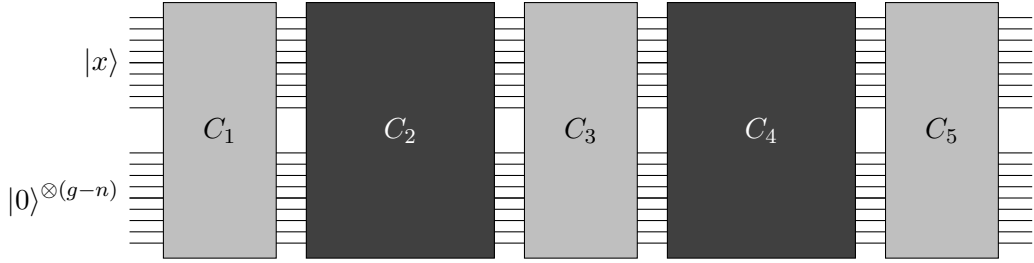
\begin{figure}[H]
    \centering
    \begin{tikzpicture}[scale=0.30,
        turn/.style={draw,
            minimum height=14mm,
            minimum width=8mm,
            fill = black,
        text=white},
        qtier/.style={draw,
            minimum height=34mm,
            minimum width=15mm,
            fill = lightgray,
        text=black},
        ctier/.style={draw,
            minimum height=34mm,
            minimum width=25mm,
            fill = darkgray,
        text=white},
        >=latex]

        \foreach \y in {-5,-4.5,...,-1} {
            \draw (-4,\y) -- (36,\y);
        }
        \foreach \y in {1,1.5,...,5} {
            \draw (-4,\y) -- (36,\y);
        }

        \draw (-4,3) node[left] {$\ket{x}$} -- (32,3);
        \draw (-4,-3) node[left] {$\ket{0}^{\otimes (g-n)}$} -- (32,-3);

        \node (Q0) at (0,0) [qtier] {$C_1$};
        \node (Q1) at (16,0) [qtier] {$C_3$};
        \node (Q2) at (32,0) [qtier] {$C_5$};

        \node (C0) at (8,0) [ctier] {$C_2$};
        \node (C0) at (24,0) [ctier] {$C_4$};
    \end{tikzpicture}
    \caption{An illustration of an $(n,5,c,q,g)$-hybrid-quantum circuit. The
    light boxes represent quantum circuits and the dark boxes represent classical
    circuits.}
    \label{fig:hybrid-quantum-circuit}
\end{figure}

\subsubsection{A Hierarchy of Hybrid Quantum Circuits}

Informally, $\HQC^i$ is the class of problems solvable by polynomial-size classical
circuits with embedded $O(\log^i(n))$-depth quantum circuits, and $\HQC$ is the
class problems solvable by polynomial-size classical circuits with embedded
$\polylog(n)$-depth quantum circuits. This notation is analogous to $\NC^i$ and
$\NC$, see Cook~\cite{Cook85}.

\begin{definition}
    $\HQC^i$ is the class of promise problems solvable by a uniform family of
    $(n,\poly(n),\poly(n),O(\log^i(n)), \poly(n))$-hybrid-quantum circuits with
    probability of error bounded by $1/3$.
\end{definition}

\begin{definition}
    $\HQC$ is the union of $\HQC^i$ over all nonnegative $i$; in symbols,
    \begin{equation}
        \HQC = \bigcup_{i \geq 0} \HQC^i.
    \end{equation}
\end{definition}

\begin{remark}\label{rem:hqc-and-bpp-bqnc}
    Observe that $\HQC$ contains $\BPP^\BQNC$ as we can restrict all the embedded
    quantum circuits in $\HQC$ to output just one bit.
\end{remark}

\subsubsection{Jozsa's Class}\label{sec:jozsas-class}

\begin{definition} \label{def:jozsacircuit}
    Define a \emph{$(n,\eta,c,q,g)$-jozsa-quantum circuit} $J$ to be a composition of
    $\eta$ circuits
    \begin{equation}
        Q_1 \circ ((\Pi \circ C_1) \otimes \mathbb{I}) \circ Q_2 \circ
        \cdots \circ Q_\eta \circ ((\Pi \circ C_\eta) \otimes \mathbb{I}) 
    \end{equation}
    such that the following hold.
    \begin{enumerate}
        \item $\Pi$ is a classical basis measurement on $g(n)/2$ qubits.
        \item $Q_1$ is an $(n, g(n), q)$-quantum tier.
        \item For all $i$, $C_i$ is an $(g(n)/2, g(n)/2, c)$-classical tier.
        \item For all $i>1$, $Q_i$ is an $(g(n), g(n), q)$-quantum tier.
        \item $g(n)$ is the width of $J$.
    \end{enumerate}
\end{definition}

\begin{definition}\label{def:jci}
    $\JC^i$ is the class of promise problems solvable by a uniform family of
    $(n, \poly(n), \poly(n), O(\log^i(n)), \poly(n))$-jozsa-quantum circuits
    with probability of error bounded by $1/3$.
\end{definition}

\begin{definition}\label{def:jcunion}
    $\JC$ is the union of $\JC^i$ over all nonnegative $i$; in symbols,
    \begin{equation}
        \JC = \bigcup_{i \geq 0} \JC^i.
    \end{equation}
\end{definition}

Define the registers \texttt{R1} and \texttt{R2} as shown in
Figure~\ref{fig:jozsa-registers}.

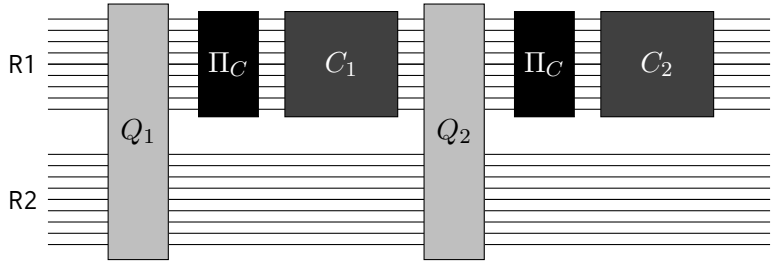
\begin{figure}[H]
    \centering
    \begin{tikzpicture}[scale=0.30,
        turn/.style={draw,
            minimum height=14mm,
            minimum width=8mm,
            fill = black,
        text=white},
        qlayer/.style={draw,
            minimum height=34mm,
            minimum width=8mm,
            fill = lightgray,
        text=black},
        measure/.style={draw,
            minimum height=14mm,
            minimum width=4mm,
            fill = black,
        text=white},
        ctier/.style={draw,
            minimum height=14mm,
            minimum width=15mm,
            fill = darkgray,
        text=white},
        >=latex]

        \foreach \y in {-5,-4.5,...,-1} {
            \draw (-4,\y) -- (28,\y);
        }
        \foreach \y in {1,1.5,...,5} {
            \draw (-4,\y) -- (28,\y);
        }

        \draw (-4,3) node[left] {\texttt{R1}} -- (28,3);
        \draw (-4,-3) node[left] {\texttt{R2}} -- (28,-3);

        \node (Q1) at (0,0) [qlayer] {$Q_1$};
        \node (Q2) at (14,0) [qlayer] {$Q_2$};

        \node (M1) at (4,3) [measure] {$\Pi_C$};
        \node (M2) at (18,3) [measure] {$\Pi_C$};

        \node (C1) at (9,3) [ctier] {$C_1$};
        \node (C2) at (23,3) [ctier] {$C_2$};
    \end{tikzpicture}
    \caption{%
        An illustration of an $(n,2,q,c)$-jozsa-quantum circuit. The light boxes
        represent quantum circuits, the black boxes represent classical basis
        measurements, and the dark boxes represent classical circuits. The width
        of the circuit $g(n)$ is split into two registers \texttt{R1} and
        \texttt{R2} of equal size.%
    }
    \label{fig:jozsa-registers}
\end{figure}

\subsubsection{Conjecture~\ref{conj:folklore}, More Formally}

With these definitions in place, we can state Conjecture~\ref{conj:folklore} as
follows.

\begin{conjecture}\label{conj:folklore-formal}
    It holds that $\HQC = \BQP$.
\end{conjecture}

\subsection{Relativized Conjecture~\ref{conj:folklore-formal}}

In this subsection, we review oracles and state a relativized version of
Conjecture~\ref{conj:folklore-formal}.

\subsubsection{Oracles in the Quantum World}\label{sec:oracl-quant-world}

For an introduction to oracles in the quantum circuit model, see Section III.4
in \cite{Watrous09}.  We recap some definitions for setting notation.

For us, an \emph{oracle} $A$ is a collection $\{A_n: n \in \NN\}$ of functions
\begin{equation}
    A_n: \{0,1\}^n \to \{0,1\}^n,
\end{equation}
to which queries can be made at unit cost.  We define
\begin{equation}
    A(x) \coloneqq A_{\abs{x}}(x)
\end{equation}
where $\abs{x}$ denotes the length of $x$.  We use the term \emph{black box} to
refer to the restriction of an oracle to inputs of a fixed length.

We represent \emph{oracle queries} by an infinite family
\begin{equation}
    \{K_n : n \in \NN\}
\end{equation}
of gates, one for each \emph{query length}.  Each gate $K_n$ is a unitary gate
acting on $n+1$ qubits, defined on the computational basis as
\begin{equation}
    K_n \ket{x}\ket{a} \mapsto \ket{x}\ket{a \oplus A(x)}
\end{equation}
where $x \in \{0,1\}^n$, $a \in \{1,0\}^n$, and $A$ is the oracle under
consideration.

\paragraph{Multiple-bit queries versus single-bit queries.} As mentioned at the
end of Section III.4 in \cite{Watrous09}, one can use the
\emph{Bernstein–Vazirani algorithm}~\cite{BernsteinV97} to simulate multiple-bit
queries with single-bit queries (after adapting the definition of $A_n$
appropriately.) Moreover, this can be performed without any non-constant
depth-overhead, so our model (after slight modifications) is equivalent to the
traditional single-bit query model.

\subsubsection{Relativized Hybrid Quantum Computation}

A \emph{relativized circuit} is one that may include query gates, and we say
that a circuit \emph{queries} a certain oracle $A$ if the outputs to its queries
are consistent with the oracle $A$. A circuit may be consistent with many
oracles; for example, a circuit that makes no queries is consistent with every
oracle. Later on, we will show that if the queries of a circuit are consistent
with a certain oracle, then it has a very low probability of success.

We define $\HQC^A$ and $\JC^A$ by replacing the classical and quantum tiers with
relativized classical and quantum tiers respectively. Put differently, we modify
the gate set for the classical and quantum tiers to include query gates to
$A$---the classical tiers make classical queries while the quantum tiers make
quantum queries.

\subsubsection{Relativized Conjecture~\ref{conj:folklore-formal}}

Relativized Conjecture~\ref{conj:folklore-formal} states that
Conjecture~\ref{conj:folklore-formal} is true relative to all oracles.

\begin{conjecture}
    \label{conj:folklore-rel}
    For all oracles $A$, it holds that $\HQC^A = \BQP^A$.
\end{conjecture}

Aaronson~\cite{Aaronson11, Aaronson14} conjectured that a weak version of
Conjecture~\ref{conj:folklore-rel} is false.

\begin{conjecture}[Aaronson~\cite{Aaronson11, Aaronson14}]
    There exists an oracle $A$ such that $\BQP^A \not\subseteq (\BPP^\BQNC)^A$.
\end{conjecture}

\subsection{The Welded Tree Problem}\label{subsec:weldedtree}

In this subsection, we review the \emph{Welded Tree Problem} of Childs et
al.~\cite{ChildsCDFGS03}. We introduce the class of graphs we will consider,
show how to turn them into black-boxes, and finally define the black-box problem
we will consider.

\subsubsection{Welded Trees}

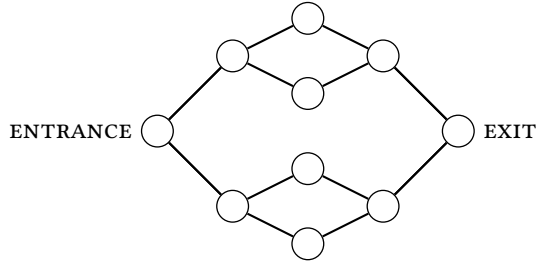
\begin{figure}[h]
    \centering
    \begin{tikzpicture}[scale=1]
        \GraphInit[vstyle=Hasse]
        \Vertex[x=0,y=0]{A}
        \node[left, xshift=-2mm] at (A){\entrance{}};
        \Vertex[x=1,y=1]{B1}
        \Vertex[x=1,y=-1]{B2}
        \Vertex[x=2,y=1.5]{C1}
        \Vertex[x=2,y=0.5]{C2}
        \Vertex[x=2,y=-0.5]{C3}
        \Vertex[x=2,y=-1.5]{C4}
        \Vertex[x=3,y=1]{D1}
        \Vertex[x=3,y=-1]{D2}
        \Vertex[x=4,y=0]{E}
        \node[right, xshift=2mm] at (E){\exit{}};
        \Edges(A,B1,C1,D1,E)
        \Edges(A,B1,C2,D1,E)
        \Edges(A,B2,C3,D2,E)
        \Edges(A,B2,C4,D2,E)
    \end{tikzpicture}
    \caption{An illustration of a $3$-welded tree.}
    \label{fig:welded-tree}
\end{figure}

\begin{definition}
    A \emph{$n$-welded tree} is a combination of two balanced binary trees $L$
    and $R$ of height $n$, with the $2^n$ leaves of $L$ identified with the
    $2^n$ leaves of $R$ in a way such that $R$ is a mirror image of $L$. For an
    illustration see Figure~\ref{fig:welded-tree}. The leftmost vertex is termed
    \entrance{} and the rightmost vertex is termed \exit{}.
\end{definition}

When $n$ is immediate from context, we will omit the $n$ and refer to the tree
as a \emph{welded tree}.

\subsubsection{Random Welded Trees}

\begin{figure}[h]
    \centering
    \begin{tikzpicture}[scale=1]
        \GraphInit[vstyle=Hasse]
        \Vertex[x=0,y=0]{A}
        \node[left, xshift=-2mm] at (A){\entrance{}};
        \Vertex[x=1,y=1]{B1}
        \Vertex[x=1,y=-1]{B2}
        \Vertex[x=2,y=1.5]{C1}
        \Vertex[x=2,y=0.5]{C2}
        \Vertex[x=2,y=-0.5]{C3}
        \Vertex[x=2,y=-1.5]{C4}
        \Vertex[x=3,y=1.5]{nC1}
        \Vertex[x=3,y=0.5]{nC2}
        \Vertex[x=3,y=-0.5]{nC3}
        \Vertex[x=3,y=-1.5]{nC4}
        \Vertex[x=4,y=1]{D1}
        \Vertex[x=4,y=-1]{D2}
        \Vertex[x=5,y=0]{E}
        \node[right, xshift=2mm] at (E){\exit{}};
        \Edges(A,B1,C1,nC4,D2,E)
        \Edges(A,B1,C2,nC1,D1,E)
        \Edges(A,B2,C3,nC3,D2,E)
        \Edges(A,B2,C4,nC2,D1,E)
        \Edges(nC4, C3)
        \Edges(nC3, C2)
        \Edges(nC1, C4)
        \Edges(nC2, C1)
    \end{tikzpicture}
    \caption{An illustration of a random $3$-welded tree $T_3$.}
    \label{fig:random-welded-tree}
\end{figure}
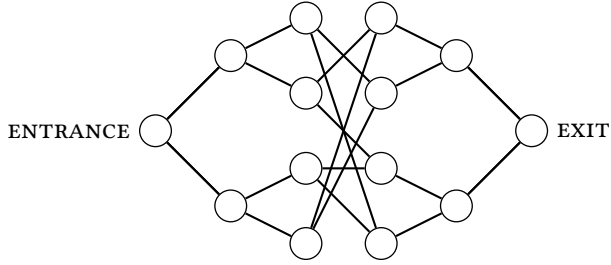

\begin{definition}
    A \emph{random $n$-welded tree} $T_n$ is a combination of two balanced
    binary trees $L$ and $R$ of height $n$ by connecting the leaves via a random
    cycle of edges which alternates between the leaves of $L$ and the leaves of
    $R$. For an illustration, see Figure~\ref{fig:random-welded-tree}.  As with
    $n$-welded trees, we term the leftmost vertex \entrance{} the rightmost
    vertex \exit{}. Notice that the \entrance{} and \exit{} vertices are
    distinguished as they are the only vertices with degree $2$.
\end{definition}

\subsubsection{Graphs with Black-Box Access}

In this paper, welded tree graphs are objects which our algorithm will only have
access to via a black-box which it can query about the neighbors of a given
vertex. To stay consistent with \cite{ChildsCDFGS03} we are also going to assume
that the graphs are edge-coloured. For this problem, we can pick an $9$-edge
coloring that does not make the problem easier (ie preserves the output
probability in expectation over colourings). Following Childs et
al.~\cite{ChildsCDFGS03}, we pick a colouring as follows.

Arbitrarily label the vertices in odd columns with colors $\{1,2,3\}$ and
arbitrarily label the vertices in even columns with colors $\{A,B,C\}$. Then
there is an induced edge coloring as follows: an edge joining an $X$-coloured
vertex to a $Y$-coloured vertex has color $XY$. For example, an edge joining a
$1$-coloured vertex and an $A$-coloured vertex has color $1A$.

\begin{definition}
    A \emph{$(n, \Xi)$-black-box graph} $G$ is a $\Xi$-edge coloured graph with
    $O(n)$ vertices whose vertices are uniquely encoded by bit strings of length
    $2n$.  We say that a $2n$-bit string is \emph{valid} with respect to $G$ if
    it is the label of a vertex in $G$.
\end{definition}

Notice that a graph may have many different corresponding black-box
graphs. Moreover, since the graph only has $O(2^n)$ vertices, $n+O(1)$ bits are
enough to give every vertex a unique label.  But we chose $2n$-bit labels so
that there are exponentially more labels than there are vertices.  Later on,
this fact is used to argue that it is hard for an adversary to guess a valid
label.

In this paper, we will only consider a restricted class of black-box graphs,
ones corresponding to random $n$-welded trees with some additional structure.

\begin{definition}
    A \emph{random $n$-welded black-box tree} $T$ is a $(n,9)$-black-box graph
    with the following additional structure.
    \begin{enumerate}
        \item $T$ is a random $n$-welded tree.
        \item The \entrance{} vertex has the label $0 \cdots 0$.
        \item The label $1 \cdots 1$ is not used for a valid vertex.  We will
            henceforth refer to this string by the name $\invalid$.
    \end{enumerate}
\end{definition}

We now define how to query a black-box.

\begin{definition}
    A \emph{query} $K_T$ to a random $n$-welded black-box tree $T$ black-box
    tree is defined as
    \begin{equation}
        K_T(x, c) \coloneqq \begin{cases}
            \text{$c$-neighbour of $x$}, & \text{$x$ is a valid vertex with a
            $c$-neighbour}\\ \invalid, & \text{otherwise}
        \end{cases}
    \end{equation}
    where $x \in \{0,1\}^{2n}$, $c$-neighbour of $x$ (the vertex joined to $x$
    by an edge with color $c$) is a $2n$-bit string, and $\invalid{} \coloneqq 1
    \cdots 1$. We define this as a unitary as
    \begin{equation}\label{eq:oracleregister}
        K_T \ket{x} \ket{c} \ket*{0^{2n}} \mapsto \ket{x} \ket{c} \ket{y}
    \end{equation}
    where $x \in \{0,1\}^{2n}$, $c \in \{1,\dots,9\}$, and $y \in
    \{0,1\}^{2n}$ is the label of the $c$-neighbour of $x$.
\end{definition}

\begin{remark}
    The quantum algorithm in Childs et al.~\cite{ChildsCDFGS03} does not make any
    queries with a superposition over the colours' register, so we can assume that
    we have $9$ unitaries---one for each colour---representing a query.
\end{remark}

\begin{definition}[querying a welded black-box tree]
    \label{def:query-black-box-tree}
    We say that a relativized circuit $C$ queries a random $n$-welded black-box
    tree $T$, denoted by $C(T)$, if all its queries $K_n$ can be replaced by
    queries $K_T$ to $T$. This is generalized to families
    \begin{equation}
        T \coloneqq \{T_n : n \in \NN\}
    \end{equation}
    of random welded black-box trees, where $T_n$ is a random $n$-welded
    black-box tree, by induction.
\end{definition}

\begin{definition}[querying a welded tree]\label{def:query-tree}
    Let $T$ be a random $n$-welded tree, and let $C(T')$ be a relativized
    circuit that queries a random $n$-welded black-box tree $T'$ corresponding
    to $T$. We define
    \begin{equation}\label{eq:output-prob}
        \Pr[C(T)] \coloneqq \dsPr_{T'}[\text{$C(T')$ returns the label of }\exit],
    \end{equation}
    where the probability is over all random $n$-welded black-box trees $T'$
    corresponding to $T$. Put differently, the probability is over all $2n$-bit
    labellings of the graph $T$. With this notation in place, we can say,
    \emph{a circuit $C(T)$ queries a random $n$-welded tree $T$} and it is
    understood that we take the output probability over all random $n$-welded
    black-box trees $T'$ corresponding to $T$.
\end{definition}

\begin{lemma}[Lemma 4 in Childs et al.~\cite{ChildsCDFGS03}]
    \label{lem:discovery-probability}
    The probability that an algorithm, which makes at most $h$ queries to a
    random $n$-welded tree, can discover the label of a vertex that was not the
    result of a query is at most
    \begin{equation}
        h \frac{2^{n+2} - 2}{2^{2n}}.
    \end{equation}
\end{lemma}
\begin{proof}
    Since we are gluing two height-$n$ binary trees each of which has $2 \cdot
    2^n$ vertices, there are $4 \cdot 2^n = 2^{n+2}$ valid labels. We know the
    location of the \entrance{} label and the label $\invalid = 1 \cdots 1$ is
    not used, so the number of unknown labels is $2^{n+2} - 2$.  Since the valid
    labels are uniformly distributed over the space of $2n$-bit strings, we get
    the desired result.
\end{proof}

\subsubsection{The Welded Tree Problem}

Given a family
\begin{equation}
    T \coloneqq \{T_n : n \in \NN\}
\end{equation}
of random welded black-box trees, where $T_n$ is a random $n$-welded black-box
tree, we define the welded tree problem relative to $T$ as follows.
\begin{table}[H]
    \centering
    $\textsc{Welded Tree Problem}(T)$\\[2mm]
    \begin{tabular}{rl}
        Input: & $0^n$ for some $n \in \NN$.\\[1mm]
        Output: & The label of the \exit{} vertex in $T_n$.
    \end{tabular}
\end{table}

\paragraph{Search versus Decision} Since the classes we want to prove lower
bounds against ($\HQC^i$ and $\JC^i$) are closed under repeating an algorithm
$O(n)$ times in parallel, the above mentioned search variant is equivalent to
the following decision variant of this problem.
\begin{table}[H]
    \centering
    $\textsc{Decision Welded Tree Problem}(T)$\\[2mm]
    \begin{tabular}{rl}
        Input: & $0^n$ for some $n \in \NN$ and $i \in \{1,\dots,n\}$.\\[1mm]
        Output: & $i$th bit of the label of the \exit{} vertex in $T_n$.
    \end{tabular}
\end{table}
So, in the remainder of the paper, we restrict our attention to the search
variant.

\paragraph{Query Length Equals Input Length} We assume that given an $n$-bit
string as input, a quantum algorithm only queries $T_n$ and not $T_m$ for any $m
\neq n$. This is without loss of generality---the idea is to replace a circuit
$Q$ with an new circuit $R$ in which all queries to $T_m$ for $m \neq n$ are
hardcoded to $\invalid$. From the description of our problem, it is immediate
that the success probability of $Q$ is no greater than $R$.

\subsubsection{Quantum Algorithm for the Welded Tree Problem}

Childs et al.~\cite{ChildsCDFGS03} gave an efficient quantum algorithm for the
\textsc{Welded Tree Problem} using quantum walks.

\begin{theorem}[Childs et al.~\cite{ChildsCDFGS03}]
    Given a family
    \begin{equation}
        T \coloneqq \{T_n : n \in \NN\}
    \end{equation}
    of random welded black-box trees, where $T_n$ is a random $n$-welded
    black-box tree, There is a quantum algorithm for the $\textsc{Welded Tree
    Problem}(T)$ which takes $\poly(n)$ time and outputs the correct answer (the
    label of the \exit{} vertex) with probability greater than $2/3$.
    Succinctly, $\textsc{Welded Tree Problem}(T) \in \BQP^T$.
\end{theorem}

\subsubsection{Classical Lower Bound for the Welded Tree Problem}

Childs et al.~\cite{ChildsCDFGS03} also gave the first classical lower bound for
the \textsc{Welded Tree Problem}, which we use as a key tool in our proof. To be
precise, we will use the following version of the lower bound, due Fenner and
Zhang~\cite{FennerZ03}, who gave an improved analysis.

\begin{theorem}[Childs et al.~\cite{ChildsCDFGS03} and Fenner and
    Zhang~\cite{FennerZ03}]
    \label{thm:class-lower-bound}
    Given a family
    \begin{equation}
        T \coloneqq \{T_n : n \in \NN\}
    \end{equation}
    of random welded black-box trees, where $T_n$ is a random $n$-welded
    black-box tree. For sufficiently large $n$, any classical algorithm for the
    $\textsc{Welded Tree Problem}(T)$ that makes at most $2^{n/3}$ queries
    outputs the correct answer (the label of the \exit{} vertex) with
    probability at most $O(n2^{-n/3})$.
\end{theorem}

\subsection{Some Definitions and Assisting Results}

\subsubsection{Distances Between States}

\begin{definition}
    Given two quantum states
    \begin{equation}
        \ket{\psi} \coloneqq \sum_x \alpha_x\ket{x}
        \quad\text{and}\quad
        \ket{\varphi} \coloneqq \sum_x \beta_x\ket{x},
    \end{equation}
    define the \emph{1-norm distance} between them as
    \begin{equation}
        \tnorm{\ket{\psi} - \ket{\varphi}}
        \coloneqq \sum_x\abs{\alpha_x - \beta_x}.
    \end{equation}
\end{definition}

\begin{definition}
    Given two probability distributions $P$ and $Q$, define the \emph{1-norm
    distance} between them as
    \begin{equation}
        \tnorm{P - Q}
        \coloneqq \sum_x \abs{P(x) - Q(x)}.
    \end{equation}
\end{definition}

\subsubsection{Intermediate Quantum States}

We are going to define a set of quantum states corresponding to the
cross-section of quantum tiers querying a $1$-random $n$-welded tree.

\begin{definition}[state at depth $\ell$]\label{def:psi-t}
    Let $Q(T)$ be a $(m,m,d)$-quantum tier, with input state $\ket{\psi_0}$, and
    querying a random $n$-welded tree $T$.  We define the state at depth $\ell$,
    denoted by $\ket{\psi_\ell}$, to be the state produced by the first $\ell$
    consecutive layers of $Q(T)$ acting on $\ket{\psi_0}$.
\end{definition}

\section{The Case of Few Tiers}\label{sec:few-tiers}

In this section, we will give a query lower bound for $(n, \eta, 4^d, d,
\poly(n))$-hybrid quantum circuits solving the welded tree problem. But our
lower bound leads to a separation against $\BQP$ only when $d \in \polylog(n)$
and $\eta < \sqrt{n}$.  In other words, this only allows us to separate
``$\FewTierHQC$'' ($\HQC$ where the hybrid quantum circuits are restricted to
have at most $\sqrt{n}$ tiers) from $\BQP$.

Our proof has two parts. First, we show that any $(n, \eta, 4^d, d, g)$-hybrid
quantum circuit can be simulated by a classical algorithm that makes at most
$4^{\eta(d+1)} (g \cdot d)$ oracle queries. Second, we combine this result with
the classical lower bound for the welded tree problem
(Theorem~\ref{thm:class-lower-bound}) to obtain the query lower bound.

The first part of the proof is formalized in the following Theorem which is the
technical heart of our proof.

\begin{theorem}\label{thm:classical-sim}
    Let $C(T)$ be an $(n, \eta, 4^d, d, g(n))$-hybrid-quantum circuit that
    queries a random $n$-welded tree $T$. Then there exists a classical
    algorithm $\A(T)$ making
    \begin{equation}
        4^{\eta(d+1)} (g(n) \cdot d)
    \end{equation}
    queries (and running in $\exp(n)$ time) such that the output probabilities of
    $C(T)$ and $\A(T)$ differ in $1$-norm error by
    \begin{equation}
        4^{(\eta+2)(d+2)} (g(n))^2 \cdot \frac{2^{n+2}-2}{2^{2n}},
    \end{equation}
    where the output probabilities of $C(T)$ and $\A(T)$ are defined over all
    possible labellings of the tree $T$.
\end{theorem}

Using Theorem~\ref{thm:classical-sim} we can prove the following adaptive
quantum query lower bound via an appeal to the classical query lower bound for
the \textsc{Welded Tree Problem} (Theorem~\ref{thm:class-lower-bound}).

\begin{theorem} \label{thm:adaptive-query-lower-bound}
    Let $T$ be a random $n$-welded black-box tree, and let $C(T)$ be a
    $(n,\eta,4^d,d,g(n))$-hybrid-quantum circuit that queries a random
    $n$-welded tree $T$, such that
    \begin{equation}
        2(\eta + 2) (d + 2) + 2\log g(n) + 2 < n/2.
    \end{equation}
    Then, for sufficiently large $n$, $C(T)$ finds the $\exit$ with probability at
    most
    \begin{equation}
        O(n2^{-n/3}).
    \end{equation}
\end{theorem}
\begin{proof}
    Using Theorem~\ref{thm:classical-sim} we can replace $C(T)$ with the
    corresponding classical algorithm $\A(T)$ that makes at most
    \begin{equation}
        4^{\eta (d+1)} (g(n) \cdot d) \leq 2^{2\eta (d+1) + 2\log g(n)}
    \end{equation}
    queries
    with at most
    \begin{align}
        4^{(\eta+2)(d+2)} (g(n))^2 \cdot \frac{2^{n+2}-2}{2^{2n}}
        &\leq \frac{2^{2(\eta+2) (d+2) + 2\log g(n) + n + 2}}{2^{2n}}\\
        &\leq \frac{2^{n/2 + n}}{2^{2n}}\\
        &\leq \frac{1}{2^{n/2}}\\
        &= O(2^{-n/2}).
    \end{align}
    loss in acceptance probability. Also, by assumption, we have
    \begin{equation}
        2^{2\eta (d+1) + 2\log g(n)} < 2^{n/2}.
    \end{equation}
    Therefore by application of the classical lower bound
    (Theorem~\ref{thm:class-lower-bound}), it follows that, for sufficiently large
    $n$, the algorithm succeeds---that is, finds the \exit{}---with probability at
    most
    \begin{equation}
        O(n2^{-n/3}) + O(2^{-n/2}) = O(n2^{-n/3}),
    \end{equation}
    as desired.
\end{proof}

\begin{theorem}[Formal version of Theorem~\ref{thm:few-tier-informal}]
    \label{thm:few-tier-formal}
    Given a family
    \begin{equation}
        T \coloneqq \{T_n : n \in \NN\}
    \end{equation}
    of random welded trees, where $T_n$ is a random $n$-welded tree. For all but
    finitely many choices of $n$, no $\FewTierHQC$ algorithm succeeds in
    deciding the $\textsc{Welded Tree Problem}(T)$ with probability higher than
    $O(n2^{-n/3})$ (where the output probability is taken over all labellings as
    defined in Definition~\ref{def:query-tree}.)
\end{theorem}
\begin{proof}
    Suppose that
    \begin{equation}
        \H(T) = \{H_n : n \in \NN \},
    \end{equation}
    is a $\FewTierHQC$ algorithm. Then, by definition, $\H(T)$ is a
    $\FewTierHQC^i$ algorithm for some $i$. More precisely, $\H(T)$ is a
    polynomial-time uniform family of
    $(n, \eta(n), 4^{d(n)}, d(n), g(n))$-hybrid-quantum circuits, for some
    functions $\eta(n) < \sqrt{n}$, $g(n) \in \poly(n)$, and
    $d(n) \in O(\log^i(n))$, querying $T$. By the assumption that $\H(T)$ is a
    polynomial-time uniform family, it follows that for all but finitely many
    choices of $H_n$, the parameters $\eta(n)$, $d(n)$, and $g(n)$ satisfy the
    relation
    \begin{equation}
        2(\eta(n) + 2) (d(n) + 2) + 2\log g(n) + 2 < n/2.
    \end{equation}
    Thus by Theorem~\ref{thm:adaptive-query-lower-bound} we get that for all but
    finitely many choices of $n$, $H_n$ succeeds with probability at most
    $O(n2^{-n/3})$. Therefore it follows that with probability $1$, $\H(T)$
    fails to decide $\textsc{Welded Tree Problem}(T)$.
\end{proof}

The remainder of this section is devoted to a proof of
Theorem~\ref{thm:classical-sim}.

\begin{definition}
    Given a $(n,\eta,4^d,d,g)$-hybrid-quantum circuit $C(T)$, for $\zeta \leq \eta$,
    let $C_\zeta(T)$ be the $\zeta$th tier of $C(T)$ (whether that tier be quantum
    or classical); for $\zeta \leq \eta$, let $C^\zeta(T)$ be the hybrid-quantum
    circuit corresponding to the first $\zeta$ tiers of $C(T)$.
\end{definition}

\subsection{Outline of Proof of Theorem~\ref{thm:classical-sim}}

Given an $(n,\eta,4^d,d,g)$-hybrid-quantum circuit $C(T)$ we define a simulation
algorithm $\A(T)$ below. First, we will need to recall a few definitions.

Let $Q(T)$ be a quantum tier in $C(T)$. For a particular input bitstring $x$ to
$Q(T)$, recall that $\ket{\psi_\ell}$ denotes the state at depth $\ell \leq d$
of $Q(T)$ as in Definition~\ref{def:psi-t}.  We wish to prove that, for every
such input $x$, the output distribution of $\A(T)$ is close to the output
distribution of $C(T)$ in trace distance.  Therefore, we fix an arbitrary input
$x$ at this point and will suppress the appearance of $x$ in our notation for
the remainder of the proof.

While it may be impossible to compute a classical description of
$\ket*{\psi_{d}}$ using only polynomially many \emph{classical} queries to $T$,
the intuition behind our classical simulation $\A(T)$ of $C(T)$ will instead be,
at each depth $\ell \leq d$, to maintain a classical description of a different
quantum state $\ket{\phi_\ell}$ which will be a close approximation of the state
$\ket{\psi_\ell}$ in trace distance.  The state $\ket{\phi_\ell}$ will be
defined inductively by the algorithm $\A(T)$ beginning with the initial
condition $\ket{\phi_0} \coloneqq \ket{\psi_0} \coloneqq \ket{x}$ and proceeding
with the simple update rule that $\A(T)$ faithfully classically simulates (with
exponential time, but just a polynomial number of classical queries) everything
that $Q(T)$ does in layer $\ell$, except for the points at which $Q(T)$ queries
the black-box $T$ at an input bitstring which is not among the "previously known
vertices" (defined later), in which case $\A(T)$ refrains from querying $T$ and
simply assumes (without justification) that the output of that query will be
$\invalid$.  As we will see below, this strategy allows $\A(T)$ to maintain a
close approximation $\ket{\phi_\ell}$ of $\ket{\psi_\ell}$ while only making a
polynomial number of \emph{classical} queries to $T$.

\subsection{The Low-Tier Simulator}
\begin{algorithm}[h]
    \tcc{Simulates $C^i(T)$ by composing the individual simulations of each of the
    first $i$ tiers of $C(T)$.} 
    \SetKwInOut{Input}{Input}\SetKwInOut{Output}{Output}
    \Input{Relativized circuit $C(T)$ and blackbox $T$}
    \Output{Simulated output of $C^i(T)$, in register $\out$; set $\Vknown$ of
    currently known vertices, in register $\vout$} 

    \tcc{initialization}
    $\Vknown \gets \text{empty dictionary}$\;
    Query the $\entrance{}$ vertex to get output $S$\;
    Set $\Vknown(\entrance{}) \gets S$\;\label{line:abuse-of-notation}
    $x \gets 0^n$\;

    \tcc{main loop}
    \For{each $j \in \{0,\dots,i\}$}{
        \uIf{$C_j$ is a quantum tier} {
            $x,\Vknown \gets \QuantumTierSimulator(C_j, x, \Vknown, T$)
            }\uElse{
            $x,\Vknown \gets \ClassicalTierSimulator(C_j, x, \Vknown, T$)
        }
    }
    \Return $x, \Vknown$
    \caption{$i$th-level ClassicalSimulationWrapper: $\A^i$}
    \label{alg:classicalsim1}
\end{algorithm}

\begin{definition}
    A \emph{dictionary} data structure is a set of key-value pairs indexed by
    keys. In other words, a dictionary $D$ has the form
    \begin{equation}
        D = \{(x_1,y_1), (x_1,y_1), \dots\}.
    \end{equation}
    We could also look at the dictionary as a mapping
    \begin{equation}
        D(x_i) \coloneqq y_i,
    \end{equation}
    for each $i$. 
\end{definition}

\begin{definition}[$\Vknown$]
    The dictionary $\Vknown$ has keys $x_i$, which are $2$-tuples
    $(v,c) \in \{0,1\}^n \times \{1,...,9\}$. We store in $\Vknown(v,c)$ the
    vertex label of the $c$-neighbour of $v$. By default, the value in
    $\Vknown(v,c)$ is $\invalid$.\footnote{One could implement this data
        structure is in a succinct way using a \emph{hash map} and an
        \texttt{if}-statement---check if the passed-in index $(v,c)$ is in the
        key-set of the hash map; if in the hash map, output the corresponding
    value; otherwise, output \invalid{}.}

    Sometimes, abusing notation, we set $\Vknown(v)$ to the output of querying a
    vertex $v$ (like in Line~\ref{line:abuse-of-notation} of
    Algorithm~\ref{alg:classicalsim1}), by this we mean that we query $(v,c)$ for
    each $c$ to get the label of the $c$-neighbour of $v$ (which can be
    $\invalid$) and then set $\Vknown(v,c)$ to be that label.
\end{definition}

\begin{definition}
    Let $\A^i(T)\vert_{\out}$ denote a modification of the algorithm $\A^i(T)$
    which only outputs the value in the $\out$ register. In other words,
    $\A^i(T)\vert_{\out}$ returns only $x$, rather than $(x, \Vknown)$.
\end{definition}

Before giving the tier simulation subroutines, we need the following definition.

\begin{definition}
    Let $L$ be a classical or quantum layer in a relativized circuit $Q(T)$. We
    can divide $L$ into two disjoint layers, one called $L^{T}$ which applies all
    of the black-box query gates in $L$ in parallel, and one called $L^{G}$, which
    applies every other gate in $L$ in parallel. Moreover, one can split $L$ into
    $L^T$ and $L^G$ in linear time.
\end{definition}

\subsection{Pseudocode for the Few-Tier Simulator}
\label{subsec:pseudocode-for-few-tier-algorithm}

The pseudocode for the low-tier simulator in given in
Algorithms~\ref{alg:SimulateOracle}, \ref{alg:QuantumLayerSimulator},
\ref{alg:QuantumTierSimulator}, and \ref{alg:ClassicalTierSimulator}.  The
analysis of the algorithm and the proof of Theorem~\ref{thm:few-tier-formal}
have been relegated to Section~\ref{sec:analysis-of-few-tier-simulator}.

\begin{algorithm}[h]
    \tcc{Simulates a black-box query by making black-box queries to vertices in
        $\Vknown$ and assuming that the output to queries made to vertices not in
    $\Vknown$ is $\invalid$.}
    \SetKwInOut{Input}{Input}\SetKwInOut{Output}{Output}
    \Input{Dictionary $\Vknown$ of known vertices, blackbox $T$, and quantum
    layer $L^T$ solely composed of query gates} 
    \Output{Array $S(z)$ of bitstrings and dictionary $\Vknown$ of currently
    known vertices} 

    Initialize $\Vknown^{temp} \gets \Vknown$\;

    \For{each bitstring $z$ which has length equal to the input register of
        $L^T$} {

        Initialize $z_{temp} \gets z$\;

        \For{each query gate $K$ in $L^T$} {

            Let $z_K$ be the substring of $z$ which lies in the input register
            of $K$\;

            Let $z_{K,x}, z_{K,c}, z_{K,y}$ be the three disjoint substrings of
            $z_K$ corresponding to the $x$-register, $c$-register, and
            $y$-register (respectively) of the input to gate $K$, as defined in
            Equation \eqref{eq:oracleregister}\;

            \uIf{$\Vknown(z_{K,x}, z_{K,c})$ exists}{

                Compute $z_{out} \coloneqq K(z)$ without any queries to $T$, by
                starting with $z_{temp}$, and replacing the substring $z_{K,y}$
                in $z$ with the substring $\Vknown(z_{K,x},z_{K,c})$\;
                Set $z_{temp} \gets z_{out}$\;

            }
            \uElseIf{$z_{K,x} == \Vknown(\alpha, \beta)$ for some $\alpha, \beta$}{

                Then, use one classical query to $T$ to set
                $\Vknown^{temp}(z_{K,x},z_{K,c}) \gets T(z_{K,x},z_{K,c})$\;
                Compute $z_{out} \coloneqq K(z)$ by starting with $z_{temp}$,
                and replacing the substring $z_{K,y}$ in $z$ with the substring
                $\Vknown^{temp}(z_{K,x}, z_{K,c})$\; Set $z_{temp} \gets
                z_{out}$\;

            }
            \Else{
                Compute $z_{out} \coloneqq K(z)$ without any queries to $T$, by
                starting with $z_{temp}$, and replacing the substring $z_{K,y}$
                in $z$ with the substring $\invalid$\;
                Set $z_{temp} \gets z_{out}$\;
            }

        }

        Set $S(z) \gets z_{temp}$\;
    }
    Set $\Vknown \gets \Vknown^{temp}$\;
    \Return $S, \Vknown$
    \caption{SimulateOracle}
    \label{alg:SimulateOracle}
\end{algorithm}

\begin{algorithm}[h]
    \SetKwInOut{Input}{Input}\SetKwInOut{Output}{Output}
    \Input{Quantum layer $L$, input quantum state $\ket{\psi}$,
    dictionary $\Vknown$ of known vertices, and blackbox $T$.}
    \Output{Simulated output of $L$ on input $\ket{\psi}$ and
    dictionary $\Vknown$ of currently known vertices.} 

    Split $L$ into a query layer $L^T$ and a non-query layer $L^G$\; Compute
    $\ket{\phi} \gets L^G\ket{\psi}$ (in exponential time) without any
    queries\label{algline:phi}\;
    Expand $\ket{\phi}$ in the classical basis as $\ket{\phi} = \sum_z c_z
    \ket{z}$\;
    Compute $S, \Vknown' \gets \SimulateOracle(\Vknown, T, L^T)$%
    \label{algline:vknown-prime}\;
    Define $\ket{\psi'} \gets \sum_z c_z \ket{S(z)}$\label{algline:psi-prime}\;
    \tcc{Note that $\ket{S(z)}$ is a bitstring for all $z$, by definition of
    $\SimulateOracle$.}

    \Return $\ket{\psi'}, \Vknown'$
    \caption{QuantumLayerSimulator}
    \label{alg:QuantumLayerSimulator}
\end{algorithm}

\begin{algorithm}[h]
    \SetKwInOut{Input}{Input}\SetKwInOut{Output}{Output}
    \Input{Relativized circuit $Q(T)$, input $x$, dictionary $\Vknown$ of known vertices,
    and blackbox $T$} 
    \Output{Simulated output of $Q(T)$ on input $x$, in register $\out$, and dictionary
    $\Vknown$ of currently known vertices, in register $\vout$} 

    Initialize $\ket{\psi_0} \gets \ket{x}$\;
    Initialize $\Vknown^0 \gets \Vknown$\;
    Let $d$ be the number of layers in $Q(T)$\;
    Let $L_i$ be the $i$th layer in $Q(T)$\;
    \For{each $i \in \{1, \dots, d \}$} {
        $\ket{\psi_i}, \Vknown^i \gets
        \QuantumLayerSimulator(L_i, \ket{\psi_{i-1}}, \Vknown^{i-1}, T)$
        \label{algline:quantumlayersimulator-output}\;
    }
    Set $\Vknown \gets \Vknown^d$\;
    Let $x$ be the output of a classical basis measurement on $\ket{\psi_d}$\;
    \Return $x, \Vknown$
    \caption{QuantumTierSimulator}
    \label{alg:QuantumTierSimulator}
\end{algorithm}

\begin{algorithm}[h]
    \SetKwInOut{Input}{Input}\SetKwInOut{Output}{Output}
    \Input{Relativized circuit $C(T)$, input $x$, set $\Vknown$ of known vertices,
    and blackbox $T$} 
    \Output{Simulated output of $C(T)$ on input $x$, in register $\out$, and set
    $\Vknown$ of currently known vertices, in register $\vout$} 

    Initialize $w_0 \gets x$\;
    Let $\eta$ be the number of layers in $C(T)$\;

    Let $L_i$ be the $i$th layer in $C(T)$\;
    \For{each $i \in \{1, \dots, \eta \}$} {
        Factorize $L_i$ into a query layer $L^T_i$ and a non-query layer $L^G_i$\;
        Compute $u_i \gets L_i^Gw_{i-1}$ without any queries\;
        Compute $w_i, \Vknown \gets \SimulateOracle(u_i, \Vknown, T, L^T)$\;       
    }
    \Return $w_\eta, \Vknown$
    \caption{ClassicalTierSimulator}
    \label{alg:ClassicalTierSimulator}
\end{algorithm}

\section{Relativized Jozsa's Conjecture}\label{sec:relat-jozs-conj}

In this section, we will give a query lower bound for $(n, 2d, c(n), d,
g(n))$-jozsa quantum circuits solving the welded tree problem. The organization
of the proof is similar to the previous section.  First, we show that any $(n,
2d, c(n), d, g)$-jozsa quantum circuit can be simulated by a classical algorithm
that makes at most $4^d + c(n)g(n)$ oracle queries. Second, we combine this
result with the classical lower bound for the welded tree problem
(Theorem~\ref{thm:class-lower-bound}) to obtain the query lower bound.

The first part of the proof is formalized in the following Theorem which is the
technical heart of our proof.

\begin{theorem}\label{thm:jozsa-classical-sim}
    Let $C(T)$ be an $(n, 2d, c(n), d, g(n))$-jozsa-quantum circuit that
    queries a random $n$-welded tree $T$. Then there exists a classical
    algorithm $\A(T)$ making at most
    \begin{equation}
        4^d + c(n) g(n)
    \end{equation}
    queries (and running in $\exp(n)$ time) such that the output probabilities of
    $C(T)$ and $\A(T)$ differ in $1$-norm error by
    \begin{equation}
        8 c(n) g(n) \cdot 4^{d+1} \cdot \frac{2^{n+2}-2}{2^{2n}}.
    \end{equation}
    where the output probabilities of $C(T)$ and $\A(T)$ are defined over all
    labellings of the tree $T$. 
\end{theorem}

Using Theorem~\ref{thm:jozsa-classical-sim} we can prove the following adaptive
quantum query lower bound via an appeal to the classical query lower bound for
the \textsc{Welded Tree Problem} (Theorem~\ref{thm:class-lower-bound}).

\begin{theorem}
    \label{thm:jozsa-main}
    Let $T$ be a random $n$-welded black-box tree, and let $C(T)$ be a $(n,
    2d, c(n), d, g(n))$-jozsa-quantum circuit that queries a $1$-random
    $n$-welded tree $T$, such that
    \begin{equation}
        2\eta d + \log c(n) + \log g(n) + 7 < n/3. 
    \end{equation}
    Then, for sufficiently large $n$, $C(T)$ finds the $\exit$ with probability at
    most
    \begin{equation}
        O(n2^{-n/3}).
    \end{equation}
\end{theorem}
\begin{proof}
    Using Theorem~\ref{thm:jozsa-classical-sim} we can replace $C(T')$ with the
    corresponding classical algorithm $\A(T')$ that makes $2^{2\eta d}$ queries
    with at most
    \begin{align}
        8 c(n) g(n) \cdot 4^{d+1} \cdot \frac{2^{n+2}-2}{2^{2n}}
    &\leq \frac{2^{3 + \log c(n) + \log g(n) + 2(d + 1) + n + 2}}{2^{2n}}\\
    &\leq \frac{2^{2d + \log c(n) + \log g(n) + 7 + n}}{2^{2n}}\\
    &\leq \frac{2^{n/2 + n}}{2^{2n}}\\
    &\leq \frac{1}{2^{n/2}}\\
    &= O(2^{-n/2}).
    \end{align}
    loss in acceptance probability. Also, by assumption, we have
    \begin{equation}
        2^{2\eta d} c(n)g(n) < 2^{n/3}.
    \end{equation}
    Therefore, by application of the classical lower bound
    (Theorem~\ref{thm:class-lower-bound}), it follows that, for sufficiently large
    $n$, the algorithm succeeds---that is, finds the \exit{}---with probability at
    most
    \begin{equation}
        O(n2^{-n/3}) + O(2^{-n/2}) = O(n2^{-n/3}),
    \end{equation}
    as desired.
\end{proof}

\begin{theorem}[Formal version of Theorem \ref{thm:jozsa-informal}]
    \label{thm:jozsa-formal}
    Given a family
    \begin{equation}
        T \coloneqq \{T_n : n \in \NN\}
    \end{equation}
    of random welded trees, where $T_n$ is a $1$-random $n$-welded tree. For
    all but finitely many choices of $n$, no $\JC$ algorithm succeeds in deciding
    the $\textsc{Welded Tree Problem}(T)$ with probability higher than
    $O(n2^{-n/3})$ (where the output probability is taken over all labellings as
    defined in Definition~\ref{def:query-tree}.)
\end{theorem}
\begin{proof}
    Suppose that
    \begin{equation}
        \J(T) = \{J_n : n \in \NN \},
    \end{equation}
    is a $\JC$ algorithm. Then, by definition, $\J(T)$ is a $\JC^i$ algorithm for
    some $i$. More precisely, $\J(T)$ is a polynomial-time uniform family of
    $(n, 2d, c(n), d \coloneqq \log^i(n), g(n))$-jozsa-quantum circuits, for some
    polynomials $c(n), g(n)$, querying $T$. By the assumption that $\J(T)$ is a
    polynomial-time uniform family, it follows that for all but finitely many
    choices of $J_n$, the parameters $c(n)$, $d \coloneqq \log^i(n)$, and $g(n)$
    satisfy the relation
    \begin{equation}
        2d + \log c(n) + \log g(n) + 7 < n/3.
    \end{equation}
    This, by Theorem~\ref{thm:jozsa-main}, implies that for all but finitely many
    choices of $n$, $J_n$ succeeds with probability at most
    $O(n2^{-n/3})$. Therefore it follows that with probability $1$, $\J(T)$ fails
    to decide $\textsc{Welded Tree Problem}(T)$.
\end{proof}

\subsection{Pseudocode for The Jozsa Simulator}

The pseudocode for the Jozsa simulator in given in
Algorithms~\ref{alg:JozsaClassicalSimulationWrapper} and
\ref{alg:QuantumLayerSimulator} which make use of algorithms defined in the
previous section. The analysis of the algorithm and the proof of
Theorem~\ref{thm:jozsa-classical-sim} have been relegated to
Section~\ref{sec:analysis-of-jozsa-simulator}.

\begin{algorithm}[h]
    \SetKwInOut{Input}{Input}\SetKwInOut{Output}{Output}
    \Input{Relativized circuit $C(T)$ and blackbox $T$}
    \Output{Simulated output of $C^i(T)$, in register $\out$; set $\Vknown$ of
    currently known vertices, in register $\vout$} 

    \tcc{initialization}
    $\Vknown \gets \text{empty dictionary}$\;
    Query the $\entrance{}$ vertex to get output $S$\;
    Set $\Vknown(\entrance{}) \gets S$\;
    $\ket{\psi_0} \gets 0^n$\;

    \Return $\JozsaQuantumTierSimulator(C_i, \ket{\psi_{i-1}}, \Vknown^{i-1}, T)$\;
    \caption{JozsaClassicalSimulationWrapper}
    \label{alg:JozsaClassicalSimulationWrapper}
\end{algorithm}

\begin{algorithm}[h]
    \SetKwInOut{Input}{Input}\SetKwInOut{Output}{Output}
    \Input{Relativized circuit $Q(T)$, input $x$, dictionary $\Vknown$ of known vertices,
    and blackbox $T$} 
    \Output{Simulated output of $Q(T)$ on input $x$, in register $\out$, and dictionary
    $\Vknown$ of currently known vertices, in register $\vout$} 

    Initialize $\ket{\psi_0} \gets \ket{x}$\;
    Initialize $\Vknown^0 \gets \Vknown$\;
    Let $d$ be the number of quantum layers in $Q(T)$\;
    Let $L_\ell$ be the $\ell$th layer in $Q(T)$\;
    \For{each $i \in \{1, \dots, \zeta \}$} {
        \If{$L_i$ is a quantum layer}{
            $\ket{\psi_i}, \Vknown^i \gets \QuantumLayerSimulator(L_i, \ket{\psi_{i-1}},
            \Vknown^{i-1}, T)$\;
        }
        \Else{
            Compute $x$ by measuring the register \texttt{R1} in the classical basis\;
            Compute $x, \Vknown^i \gets \ClassicalTierSimulator(L_i,x,\Vknown^{i-1},
            T)$\;
            $\ket{\psi_i} \gets \ket{x}\ket{\psi_{i-1}^{\texttt{R2}}}$ where
            $\ket{\psi_{i-1}^{\texttt{R2}}}$ is the quantum state on register \texttt{R2}\;
        }
    }
    Set $\Vknown \gets \Vknown^{2d}$\;
    Let $x$ be the output of a classical basis measurement on $\ket{\psi_{\eta}}$\;
    \Return $x, \Vknown$
    \caption{JozsaQuantumTierSimulator}
    \label{alg:JozsaQuantumTierSimulator}
\end{algorithm}

\section{The Information Bottleneck, and the Case of Polynomial Tiers.}
\label{sec:hqc}

In this section we will introduce a proof technique, which we refer to as the
``Information Bottleneck'', which allows us to extend the results of the
previous sections to our main result about $\HQC$.

\begin{theorem}[Formal Version of Theorem~\ref{thm:mainthm-informal}]
    \label{thm:mainthm-formal}
    Given a family
    \begin{equation}
        T \coloneqq \{T_n : n \in \NN\}
    \end{equation}
    of random welded trees, where $T_n$ is a random $n$-welded tree. For all but
    finitely many choices of $n$, no $\HQC$ algorithm succeeds in deciding the
    $\textsc{Welded Tree Problem}(T)$ with probability higher than
    $O(2^{-\Omega(n)})$ (where the output probability is taken over all
    labellings as defined in Definition~\ref{def:query-tree}.)
\end{theorem}

We begin by noting that, in the setting of $\poly(n)$ tiers, we can assume,
without loss of generality, that every tier is a quantum tier because a
classical tier of polynomial depth can always be implemented as the composition
of polynomially many quantum tiers of logarithmic depth.  Recall that the reason
that this does not necessarily contain all of $\BQP$ is that there is a required
measurement of all qubits, in the computational basis, at the end of each tier.
For simplicity of notation we will take this interpretation for the remainder of
this section.  So, we will use the following modification of Definition
\ref{def:hybridcircuit}, which we can use, without loss of generality, in this
setting of $\poly(n)$ tiers.

\begin{definition}\label{def:wlog-hybrid-quantum-circuit}
    Define a \emph{$(n,\eta,q,g)$-hybrid-quantum circuit} to be a composition of
    $\eta$ circuits
    \begin{equation}
        C_1 \circ C_2 \circ \cdots \circ C_\eta
    \end{equation}
    such that the following hold.
    \begin{enumerate}
        \item $C_1$ is an $(n, g, q)$-quantum tier.
        \item $C_\eta$ has at least one output.
        \item For $i > 1$, $C_i$ is an $(g, g, q)$-quantum tier.
    \end{enumerate}
\end{definition}

The first step in proving the oracle separation is to augment our previous
algorithm for classically simulating a relativized $\HQC$ circuit by adding a
subroutine called $\bottleneck$ which limits the growth of the set $\Vknown$ to
be polynomial in the number of tiers rather than exponential.  The challenge is
to do this while also preserving the properties of $\Vknown$ required for the
rest of the algorithm to work.  As we will see in the analysis of this algorithm
in Section~\ref{sec:analysis-of-bottleneck-simulator}, the key property of
$\Vknown$ that we wish to preserve is the property that, at any point in our
simulation algorithm, it is impossible to guess a valid label outside of
$\Vknown$ with better than exponentially small ($2^{-n/100}$) success
probability.

\subsection{The Bottleneck Algorithm}

In this Subsection, we will discuss the algorithm and provide pseudocode for the
simulator. The analysis of the algorithm and the proof of
Theorem~\ref{thm:mainthm-formal} have been moved to
Section~\ref{sec:analysis-of-bottleneck-simulator}.

Given a \emph{$(n,\eta, q(n), g(n))$-hybrid-quantum circuit} $C(\cdot)$ where
$\eta = \poly(n)$, $q(n) = \polylog(n)$, and $g(n) = \poly(n)$, the outer loop
for our new algorithm for simulating $C(T)$ is given in
Algorithm~\ref{alg:bottleneck-classicalsim}.

\begin{algorithm}[h]
    \tcc{Simulates $C^i(T)$ by composing the individual simulations of each of
    the first $i$ tiers of $C(T)$.}
    \SetKwInOut{Input}{Input}\SetKwInOut{Output}{Output}
    \Input{Relativized circuit $C(T)$ and blackbox $T$}
    \Output{Simulated output of $C^i(T)$, in register $\out$; set $\Vknown$ of
    currently known vertices, in register $\vout$}

    \tcc{base case}
    \If{$i == 0$}{
        Initialize $\Vknown \gets \text{empty dictionary}$\;
        Query the $\entrance{}$ vertex to get output $S$\;
        Set $\Vknown(\entrance{}) \gets S$\;
        Initialize $\Vhknown \gets \Vknown$\;
        $x \gets 0^n$\;
        \Return $x, \Vknown, \Vknown^{hist}$
    }

    \tcc{inductive case}
    $y, \Viknown, \Vknown^{hist} \gets \M^{i-1}(C(T), T)$\;
    $x,\Vfknown, \Vknown^{hist} \gets
    \BottleneckQuantumTierSimulator(C_i, i-1,y, \Viknown, \Vknown^{hist}, T$)\;
    $\Vknown \gets \Vfknown$\;

    \Return $x, \Vknown, \Vknown^{hist}$
    \caption{$i$th-level ClassicalSimulationWrapper with Bottleneck: $\M^i$}
    \label{alg:bottleneck-classicalsim}
\end{algorithm}

The first thing to notice about the new simulation algorithm $\M^i$ is that it
now tracks two different sets of ``known vertices'', one called $\Vknown$, and
another called $\Vhknown$.  At any point in this new algorithm $\Vhknown$ will
contain the label and neighbors of every vertex ever queried in the course of
the algorithm up to that point.  This will serve as an important reference
throughout the proof.  The set $\Vknown$, on the other hand, will be modified by
$\M^i$, after each simulated tier, to contain just the vertex labels which could
be guessed to be valid with probability higher than $2^{-n/100}$ given the
output of that tier (plus a select few other labels for technical reasons).
Then, when algorithm $\M^i$ begins to simulate each quantum layer, it only
treats the vertex labels in $\Vknown$ as known, rather than all of the labels in
the larger set $\Vhknown$.  This modification drastically decreases the amount
of ``branching'' that $\M^i$ does while simulating a quantum tier, and this is
what allows us to keep the size of $\Vknown$ (and $\Vhknown$) from growing too
large during the course of $\M^i$.  But, the process of updating the new
$\Vknown$ is delicate, and is carried out by the subroutine $\bottleneck$
(Algorithm~\ref{alg:Bottleneck}), which is itself a subroutine of the subroutine
$\BottleneckQuantumTierSimulator$
(Algorithm~\ref{alg:BottleneckQuantumTierSimulator}) in $\M^i$.

\begin{algorithm}[h]
    \SetKwInOut{Input}{Input}\SetKwInOut{Output}{Output}
    \Input{Relativized tier $Q(T)$, number $j$ of the current tier, input $x$,
        dictionary $\Vknown$ of currently known vertices, dictionary
        $\Vknown^{hist}$ of all vertices ever encountered, and
    blackbox $T$} 
    \Output{Simulated output of $Q(T)$ on input $x$, in register $\out$,
        dictionary $\Vknown$ of currently known vertices, in register $\vout$, and
    dictionary $\Vhknown$ of every vertex ever queried, in register $\vhout$.} 

    Initialize $\ket{\psi_0} \gets \ket{x}$\;
    Initialize $\Vknown^0 \gets \bottleneck(j-1, x, \Vhknown, \emptydict)$\;
    \tcc{Here $\emptydict$ represents the empty dictionary}
    Let $d$ be the number of layers in $Q(T)$\;
    Let $L_\ell$ be the $\ell$th layer in $Q(T)$\;
    \For{each $\ell \in \{1, \dots, d \}$} {

        $\ket{\psi_\ell}, \Vknown^{\text{$\ell$-temp}} \gets
        \QuantumLayerSimulator(L_\ell, \psi_{\ell-1}, \Vknown^{\ell-1},T)$%
        \label{algline:newvknown-i}\;

        Update $\Vknown^{hist} \gets \merge(\Vknown^{hist},
        \Vknown^{\ell\text{-temp}})$\;
        \tcc{Here $\merge$ performs a standard concatenation of two dictionaries.}

        $\Vknown^{\ell} \gets \bottleneck(j-1, x, \Vknown^{hist}, \Vknown^{\ell\text{-temp}})$\;

    }
    Set $\Vknown \gets \Vknown^d$\;
    Update $\Vknown^{hist} \gets \merge(\Vknown^{hist}, \Vknown)$\;
    Let $x$ be the output of a classical basis measurement on $\ket{\psi_d}$\;
    \Return $x, \Vknown, \Vknown^{hist}$
    \caption{BottleneckQuantumTierSimulator}
    \label{alg:BottleneckQuantumTierSimulator}
\end{algorithm}

One similarity between the two simulators $\A^i$ and $\M^i$ is that the only
randomness used in either of them (or in all of $\M^i$) is when they sample from
the computational basis elements $\ket{z}$ of a quantum state $\ket{\psi} =
\sum_{z}c_z \ket{z}$ according to the probability distribution given by
$|c_z|^2$.  Since the number of bits in $\ket{z}$ is at most $g(n)$, this
process can be done with $ng(n)$ random bits, approximately, but to
exponentially good precision in $n$.  For conciseness, in our analysis we will
ignore this exponentially small error in sampling as it can easily be included
in our union bounds, we will simply assume that the sampling is perfect.  This
step can easily be justified by standard techniques. Since this sampling process
is only done once per quantum layer, of which there are only $q(n)$ per quantum
tier, it follows that only $n\eta q(n)g(n)$ random bits are required for the
entire algorithm $\M^\eta$.  For concreteness in the rest of the argument we
name this seed randomness $r$ ($\abs{r} \leq n\eta q(n)g(n)$), and we consider
$\M^{i}$ to be a deterministic algorithm which is a function of $r$.  Note that
$\M^i$ only actually uses the first $n\eta g(n)\cdot i$ bits of $r$, which we
will denote by $r_{\leq i}$.  When necessary we will use the notation
$\M^i_{r_{\leq i}}$ to highlight this, although we will omit the $r_{\leq i}$
when it is not relevant.

The key subroutine in $\BottleneckQuantumTierSimulator$ is called $\bottleneck$.
Unlike every other subroutine discussed so far, $\bottleneck$ is able to safely
reduce the size of the estimated size of the effective set of explored vertices
$\Vknown$ by leveraging the ``Information Bottleneck" principle to show that
some vertex labels in $\Vknown$ are not, in fact, well correlated with the
output of the current quantum tier.  Since the labels of those vertices cannot
be guessed given the status of the $\HQC$ circuit at this tier (except with
exponentially small probability) they can be safely removed from $\Vknown$.

\begin{definition}
    For $a \in \{1,2,3\}$ we will let $\M^i(C(T), T)[a]$ denote
    the $a^{th}$ element of the tuple $x, \Vknown, \Vhknown$ output by $\M^i(C(T),
    T)$.
\end{definition}

\begin{definition}\label{def:TVsr}
    For any $n$, dictionary $V$, and bitstring $s \in \{0,1\}^{g(n)}$, let
    \begin{equation}
        \begin{split}
            \T_{V, s,r_{\leq i}}^i \coloneqq
            \{
            &\text{random $n$-welded black-box trees P such that}\\
            &s = \M^i_{r_{\leq i}}(C(P), P)[1]\\
            &\text{and P is consistent with $V$}
            \}.
        \end{split}
    \end{equation}
    When we say $P$ is consistent with $V$, we mean that the black-box tree $P$
    and the dictionary $V$ agree on all the labels, colors, and adjacencies
    specified by $V$.  Note that $\T_{V, s,r_{\leq i}}^i$ is a well defined set
    because $\M^i_{r_{\leq i}}(C(P), P)[1]$ is a deterministic function of $P$
    (for the fixed random seed $r$).

    Let
    \begin{equation}
        \begin{split}
            \T_{V} \coloneqq
            \{\text{random $n$-welded black-box trees $P$}\phantom{\}.}\\
            \text{such that $P$ is consistent with $V$}\}.
        \end{split}
    \end{equation}
    Further, let $\T = \cup_s \T_{\emptydict, s, r}$ denote the set of all
    random $n$-welded black-box trees $P$. Note that every blackbox P is
    consistent with the empty dictionary $V = \emptydict$, so $ \cup_s
    \T_{\emptydict, s, r} = \T$.
\end{definition}

\begin{definition}
    For a set of random $n$-welded black-box trees $\mathcal{G}$ we let
    \begin{equation}
        \mathbb{P}_{P \in \mathcal{G}} \bigl[
            \text{$b$ is a valid label in $P$}
        \bigr]
    \end{equation}
    denote the probability that $b$ is a valid label in $P$ when $P$ is selected
    uniformly at random from the set $\mathcal{G}$.
\end{definition}

\begin{algorithm}[H]
    \SetKwInOut{Input}{Input}\SetKwInOut{Output}{Output}
    \Input{Index $i$ of current quantum tier, bit string $x$, dictionaries
        $\Vcknown\subseteq \Vhknown$ of initially known and finally known
    vertices respectively}
    \Output{Dictionary $\Vknown$ of ``effectively known'' vertices, satisfying
    $\Vcknown\subseteq \Vknown \subseteq \Vhknown$}

    \If{
        $\abs{\T_{\Vcknown,x,r_{\leq i}}^i} < 2^{-n(g(n)+\abs{r})}
        \abs{\T_{\Vcknown}}$ \label{algline:haltifsmall1}
        }{

        $\abort$ and guess a random label for the $\exit$ vertex of the entire
        welded tree problem on $T$\label{algline:haltifsmall2}\;

    }

    Initialize $\Vknown \gets \Vcknown$\;

    Our goal is to build $\Vknown$ into a set satisfying $\Vcknown\subseteq
    \Vknown \subseteq \Vhknown$, and\;

    \[
        \begin{split}
            \forall b \in \{0,1\}^{2n}
            \text{ such that $b$ does not appear in }\Vknown:\\
            \mathbb{P}_{P \in \T_{\Vknown, x, r_{\leq i}}^i}
            \left [b \text{ is a valid label in P} \right ]
            \leq 2^{-n/100}
        \end{split}
    \]\label{eq:badguess2}

    \tcc{Note that one can compute the set  $\T_{\Vknown, x, r_{\leq i}}^i$ in
        doubly exponential time without using any queries to $T$ (in fact that
        set has nothing to do with $T$).  We can then compute the LHS of
        Equation in line~\ref{eq:badguess2} using the same resources.  With this
        ability we can use the following greedy algorithm to add labels to
    $\Vknown$:}

    \While{Equation in line~\ref{eq:badguess2} \text{ is not
        true }}{

        Compute an arbitrary $b'$ violating Equation in line~\ref{eq:badguess2}
        (This can be done in doubly exponential time)\;

        \If{$b'$ does not appear in $\Vhknown$ \label{algline:guessabort1}}{

            $\abort$ and guess a random label for the $\exit$ vertex of the
            entire welded tree problem on $T$\label{algline:guessabort2}\;

        }

        If $b'$ does appear in $\Vhknown$, then add $b'$ and its children, and
        edge colors in $\Vhknown$ to the dictionary $\Vknown$, and continue\;

    }

    \tcc{After the above while loop terminates we conclude the subroutine with
    the following clean-up step.}

    Let $\Vknown^{\text{complete}} \subseteq \Vhknown$ be the minimum size
    subtree (rooted at $\entrance$) of $\Vhknown$ which contains $\Vknown$.\;

    \tcc{Since $\Vknown \subseteq \Vhknown$ and $\Vhknown$ is a tree rooted at
        $\entrance$ we can compute $\Vknown^{\text{complete}}$ without any
        queries to $T$, only look-ups to $\Vhknown$.  We will see in the
        analysis that this does not adversely increase the size of $\Vknown$.}

    $\Vknown \gets \Vknown^{\text{complete}}$\label{algline:extendtotree}\;

    \Return $\Vknown$
    \caption{Bottleneck}
    \label{alg:Bottleneck}
\end{algorithm}

\section{Open Problems}

We wonder if the black-boxes used in our proof can be constructed based on
cryptographic assumptions.

\begin{problem}
    Assuming post-quantum classical indistinguishability obfuscation, is it
    possible to construct an explicit family
    \begin{equation}
        \{T_n : n \in \NN\}
    \end{equation}
    of Welded Tree black boxes such that they separate $\BQP$ from $\HQC$?
\end{problem}

Further, we wonder if it is possible to quantum-secure-VBB obfuscate these
black boxes. 

\begin{problem}
    Is it possible to construct an explicit family
    \begin{equation}
        \{T_n : n \in \NN\}
    \end{equation}
    of Welded Tree black boxes that can be quantum-secure-VBB obfuscated, such
    that they separate $\BQP$ from $\HQC$?
\end{problem}

\subsection*{Acknowledgements}

We thank Richard Cleve, Aram Harrow, John Watrous, and Umesh
Vazirani for helpful comments and discussions.

MC was supported at the IQC by Canada's NSERC and the Canadian
Institute for Advanced Research (CIFAR), and through funding
provided to IQC by the Government of Canada and the Province of
Ontario.  This work was completed at the IQC/University of Waterloo.

\bibliography{jozsa}

\newpage
\appendix

\section{Analysis of the Few-Tier Simulator}
\label{sec:analysis-of-few-tier-simulator}

We will proceed in a modular fashion, beginning with the correctness of
$\QuantumTierSimulator$, followed by $\QuantumTierSimulator$ and
$\ClassicalTierSimulator$, and ending with the $i$th-level
$\ClassicalSimulationWrapper$ $\A^i$.

\begin{lemma}
    \label{lem:quantumlayersimulator-works}
    Suppose $L$ is a quantum layer. Let $\ket{\psi}, \ket{\psi'}$ and $\Vknown,
    \Vknown'$ as in Algorithm~\ref{alg:QuantumLayerSimulator}, and let $g(n)$ be
    the width of $L$. Further, assume that $g(n)4^d\abs{\Vknown} = \subexp(n)$.
    Then it holds that
    \begin{equation}\label{eq:updated-vknown}
        \bigabs{\Vknown'} \leq 4 \bigabs{\Vknown}
    \end{equation}
    and
    \begin{equation}\label{eq:updated-tnorm}
        \tnorm{ \ket*{\psi'} - L\ket{\psi} }
        \leq 4 g(n)\abs{\Vknown'} \cdot \frac{2^{n+2}-2}{2^{2n}}.
    \end{equation}
\end{lemma}
\begin{proof}
    Let $L^G$ and $L^T$ be as in Algorithm~\ref{alg:QuantumLayerSimulator}.
    Lets start by proving Equation~\eqref{eq:updated-vknown}. Recall that
    vertices in a random welded tree have degree at most $3$. Therefore, since
    the $\SimulateOracle( \Vknown, T, L^T)$ subroutine queries, at most, every
    vertex in $\Vknown$, we know that the new set $\Vknown'$ of known vertices
    has at most $3\abs{\Vknown}$ new vertices, plus the original $\abs{\Vknown}$
    vertices that were already contained in $\Vknown$ itself (since $\Vknown'
    \subseteq \Vknown$ by definition). Thus it holds that
    \begin{equation}\label{eq:bound-on-vknown-size}
        \abs*{\Vknown'} \leq 4 \abs{\Vknown}.
    \end{equation}

    Now, lets move on to proving Equation~\eqref{eq:updated-tnorm}. As defined in
    Line~\ref{algline:phi} of the Algorithm~\ref{alg:QuantumLayerSimulator}, we
    have that $\ket{\phi} = L^G \ket{\psi}$.  Recall that we are given an
    exponential-size classical description of $\ket{\psi}$---this, along with the
    fact that $L^G$ only applies standard quantum gates (i.e., no black-box
    queries), allows us to compute $\ket{\phi}$ in exponential time without using
    any queries to $T$.

    The harder step is computing an approximation to $L^T\ket{\phi}$. We do this
    by making use of Algorithm~\ref{alg:SimulateOracle} as follows
    \begin{equation}
        S, \Vknown' \gets \SimulateOracle(\Vknown, T, L^T),
    \end{equation}
    and then setting
    \begin{equation}
        \ket*{\psi'} \gets \sum_z c_z \ket{S(z)}.
    \end{equation}
    Define the set
    \begin{equation}\label{eq:outliers}
        \Outliers \coloneqq \{z: \ket{S(z)} \neq L^T \ket{z}\}.
    \end{equation}
    and notice that if $\ket{S(z)} \neq L^T \ket{z}$ then they are unequal
    classical basis states and, therefore, perpendicular.  Let us decompose
    $\ket{\phi} = \sum_z c_z \ket{z}$. Making use of \eqref{eq:outliers} we can
    restate the fidelity between the simulated state and the true state as
    \begin{align}
        \fid(\ket*{\psi'}, L^T\ket*{\phi})
    &= \mel*{\psi'}{L^T}{\phi}\\
    &= \sum_z \abs{c_z}^2 \mel*{S(z)}{L^T}{z}\\
    &= \sum_{z \notin \Outliers} \abs{c_z}^2\\
    &= 1 - \sum_{z \in \Outliers} \abs{c_z}^2.
    \end{align}
    Thus, by applying the \emph{Fuchs-van de Graaf inequalities}~\cite{FuchsG99}
    we get
    \begin{align}
        \tnorm{\ket*{\psi'} - L^T\ket{\phi}}
    &\leq 2 \sqrt{\fid(\ket{\psi'},L^T\ket{\phi} )}\\
    &= 2 \sqrt{\sum_{z \in \Outliers} \abs{c_z}^2}.\label{eq:boundbyset}
    \end{align}

    We will now use Equation \eqref{eq:boundbyset} to argue that
    $\ket{\psi_{\ell+1}}$, and $L_{\ell+1}^T\ket{\phi_{\ell+1}}$ must be very
    close to the same state.

    Consider, as a thought experiment, a classical algorithm $\mathcal{B}$ which
    begins with the classical description of $\ket{\phi'} = \sum_z c_z \ket{z}$
    and attempts to guess a valid vertex which is not contained in $\Vknown$ by
    first sampling a random $z$ with probability $|c_z|^2$, and then randomly
    picking a query gate $K$ in $L^T$ and returning the substring of $z$ which
    lies in the input register of $K$.

    By the definition of the subroutine $\SimulateOracle$, the set $\Outliers$ contains
    those $z$ which, in the input register of at least one query gate $K$ in
    $L^T$, have a substring which is the label of a valid vertex outside of
    $\Vknown'$.  Thus, if $\mathcal{B}$ successfully guesses a $z \in \Outliers$,
    which happens with probability $\sum_{z \in \Outliers} \abs{c_z}^2$, and
    further happens to guess the correct $K$ in $L^T$, which happens with
    probability at least $1/g(n)$ (recall that $g(n)$ is the width of $L$) then
    $\mathcal{B}$ has successfully guessed the substring of $z$ in the input of
    register to $K$, which is a valid vertex that is not contained in
    $\Vknown'$.  Therefore, the success probability of $\mathcal{B}$ is at
    least
    \begin{equation}
        \frac{1}{g(n)} \cdot \sum_{z \in \Outliers} \abs{c_z}^2.
    \end{equation}  
    However, the algorithm $\mathcal{B}$ only uses at most $\abs*{\Vknown'}$
    classical queries to $T$ because it only uses a classical description of
    $\ket{\phi_{\ell+1}}$ (which, by definition, can be computed with
    $\abs*{\Vknown'}$ queries). By Lemma~\ref{lem:discovery-probability} the
    chance that a classical algorithm which makes at most $h$ queries guesses a
    valid vertex not returned by the oracle is at most
    \begin{equation}\label{eq:classical-chance-of-guesing}
        h \frac{2^{n+2} - 2}{2^{2n}}.
    \end{equation}
    Therefore, it follows that
    \begin{align}\label{eq:boundbyguess}
        \frac{1}{g(n)} \sum_{z \in \Outliers} |c_z|^2
        \leq \Pr[\mathcal{B}\text{ succeeds} ]
        \leq \abs*{\Vknown'} \cdot \frac{2^{n+2}-2}{2^{2n}} 
    \end{align}
    Combining Equation \eqref{eq:boundbyguess} with Equation
    \eqref{eq:boundbyset} leads to
    \begin{align}
        \tnorm{\ket{\psi'} - L^TL^G\ket{\psi}}
    &=\tnorm{\ket{\psi'} - L^T\ket{\phi}}\\
    &\leq 4 \sqrt{
        g(n)\abs*{\Vknown'} \cdot \frac{2^{n+2}-2}{2^{2n}}
    }.
    \end{align}
    But recall that we assumed that $g(n)4^d\abs{\Vknown} = \subexp(n)$ and we
    proved that $\abs*{\Vknown'} \leq 4\abs{\Vknown}$
    (Equation~\eqref{eq:bound-on-vknown-size}). By definition, we have
    \begin{equation}
        \frac{2^{n+2}-2}{2^{2n}} = \frac{1}{\exp(n)}
    \end{equation}
    so it holds that
    \begin{equation}
        g(n)\abs*{\Vknown'} \cdot \frac{2^{n+2}-2}{2^{2n}} \leq 1
    \end{equation}
    for sufficiently large $n$. Therefore, we get
    \begin{equation}
        \label{eq:boundbyguess2} 
        \tnorm{\ket*{\psi'} - L^TL^G\ket{\psi}}
        \leq 4 g(n)\abs*{\Vknown'} \cdot \frac{2^{n+2}-2}{2^{2n}}
    \end{equation}
    as desired.
\end{proof}

\begin{lemma} \label{lem:quantum-composition}
    Let $(x, \Vknown)$ be the random variable produced by $\A^\zeta(T)$, for some
    $\zeta \in \{1,\dots,\eta\}$. Say $d$ is the depth of $C_{\zeta+1}$, $g(n)$ is
    the width of $C_{\zeta+1}$, and $\abs{\Vknown}$ is the number of vertices in
    $\Vknown$. Further, assume that
    \begin{equation}\label{eq:quantum-composition-assumption}
        g(n) 4^d \abs{\Vknown} = \subexp(n). 
    \end{equation}
    Then the following statements hold.
    \begin{enumerate}
        \item It holds that\label{item:1}
            \begin{equation}
                \begin{split}
                    \tnorm{ \QuantumTierSimulator(C_{\zeta+1}, x, \Vknown, T)\vert_\out -
                    C_{\zeta+1}(x) }\\
                    \leq 4 g(n) \cdot 4^{d+1} \abs{\Vknown} \cdot \frac{2^{n+2}-2}{2^{2n}}.
                \end{split}
            \end{equation}
        \item $\QuantumTierSimulator(C_{\zeta+1}, x, \Vknown, T)$ makes at most
            $4^d \cdot \abs{\Vknown}$ queries to $T$.\label{item:2}
    \end{enumerate}
\end{lemma}
\begin{proof}
    Let's denote by $\mathcal{C}[\ell]$ the intermediate quantum state produced
    at the end of the first $\ell$ layers of a circuit $\mathcal{C}$. For $\ell
    \in \{1,\dots,\eta\}$, let $\ket{\psi_{\ell}}$ and $\Vknown^\ell$ be as
    defined in Line~\ref{algline:quantumlayersimulator-output} of the
    $\QuantumTierSimulator$ subroutine.  So, $\Vknown^{0}$ is identical to the
    set $\Vknown$ which is input to the $\QuantumTierSimulator$ subroutine.

    First, we will show, by induction, that at every depth $\ell \leq d$, it holds
    that
    \begin{equation}\label{eq:vknown-at-each-depth}
        \abs*{\Vknown^{\ell}} \leq 4^{\ell}  \abs{\Vknown}
    \end{equation}
    and
    \begin{equation}\label{eq:tnorm-at-each-depth}
        \tnorm{ \ket{\psi_{\ell}} - C_{i+1}[\ell] }
        \leq 4 g(n) \cdot \left(\sum_{i=0}^{\ell}\abs*{\Vknown^i}\right) \cdot
        \frac{2^{n+2}-2}{2^{2n}}, 
    \end{equation}
    Since $\abs*{\Vknown^d}$ is the number of known vertices at the end of the
    $\QuantumTierSimulator$ subroutine, we can simulate this algorithm (without
    any error) by an algorithm that makes at most $\abs*{\Vknown^d}$ black-box
    queries. This gives us Part~\ref{item:2} of the lemma.  Further, combining
    \eqref{eq:tnorm-at-each-depth} with \eqref{eq:vknown-at-each-depth}, and
    using the sum of geometric series we get Part~\ref{item:1} of the lemma.

    \paragraph{Base Case.} Suppose $\ell=0$, then it is immediate that
    Equation~\eqref{eq:tnorm-at-each-depth} holds because, by definition,
    $\ket{\psi_{0}} = C_{i+1}[0] = \ket{x}$, and
    Equation~\eqref{eq:vknown-at-each-depth} holds because, by definition,
    $\abs*{\Vknown^{0}} = 4^{0} \abs*{\Vknown^{0}}$.

    \paragraph{Inductive Case.} Suppose that, for a given $\ell \geq 0$,
    Equations~\eqref{eq:vknown-at-each-depth} and \eqref{eq:tnorm-at-each-depth}
    hold.  We will now prove that it must also hold for $\ell + 1$.
    Equation~\eqref{eq:vknown-at-each-depth} for $\ell + 1$ follows immediately
    from Lemma~\ref{lem:quantumlayersimulator-works}.

    Now, let's prove Equation~\eqref{eq:vknown-at-each-depth} for $\ell +
    1$. Let's start with the triangle inequality
    \begin{align}
        \tnorm{\ket{\psi_{\ell+1}} - C_{i+1}[\ell+1]}
        \leq \tnorm{\ket{\psi_{\ell+1}} -
        L_{\ell+1}^TL_{\ell+1}^G\ket{\psi_{\ell}}}
        + \tnorm{L_{\ell+1}^TL_{\ell+1}^G\ket{\psi_{\ell}} -
        L_{\ell+1}^TL_{\ell+1}^GC_{i+1}[\ell] },
    \end{align}
    and use the fact that the trace norm is nonincreasing under quantum channels
    (see Theorem 9.2 in Nielsen and Chuang~\cite{NielsenC10}) to get
    \begin{align}
        \tnorm{\ket{\psi_{\ell+1}} - C_{i+1}[\ell+1]}
        \leq \tnorm{\ket{\psi_{\ell+1}} -
        L_{\ell+1}^TL_{\ell+1}^G\ket{\psi_{\ell}}}
        + \tnorm{\ket{\psi_{\ell}} - C_{i+1}[\ell] }.
    \end{align}
    Lemma~\ref{lem:quantumlayersimulator-works} and the inductive assumption
    (Equation~\eqref{eq:tnorm-at-each-depth}) further simplifies this to
    \begin{align}
        \tnorm{\ket{\psi_{\ell+1}} - C_{i+1}[\ell+1]}
        &\leq  4 g(n)\abs*{\Vknown^{\ell+1}} \cdot \frac{2^{n+2}-2}{2^{2n}}
        + 4 g(n)\cdot \left(\sum_{i=0}^{\ell}\abs*{\Vknown^i}\right)
        \frac{2^{n+2}-2}{2^{2n}} \\
        &\leq 4 g(n)\cdot \left(\sum_{i=0}^{\ell+1}\abs*{\Vknown^i}\right)
        \frac{2^{n+2}-2}{2^{2n}}.
    \end{align}
    This proves that Equation~\eqref{eq:tnorm-at-each-depth} holds for $\ell + 1$.
    This completes Parts~\ref{item:1} and \ref{item:2} of the lemma.
\end{proof}

\begin{lemma}\label{lem:classical-composition}
    Let $(x, \Vknown)$ be the random variable produced by $\A^i(T)$. Suppose $d$
    is the depth of $C_{i+1}$, $g(n)$ the width of $C_{i+1}$, and $\abs{\Vknown}$
    is the number of vertices in $\Vknown$.
    \begin{enumerate}
        \item Then it holds that
            \begin{equation}
                \tnorm{\ClassicalTierSimulator(C_{i+1}, x, \Vknown, T)\vert_\out
                - C_{i+1}(x)}
                \leq g(n) \cdot d\abs{\Vknown} \cdot \frac{2^{n+2}-2}{2^{2n}}.
            \end{equation}
        \item $\ClassicalTierSimulator(C_{i+1}, x, \Vknown, T)$ makes at most
            $g(n) \cdot d$ queries.
    \end{enumerate}
\end{lemma}
\begin{proof}
    Since $g(n) \cdot d$ is the circuit size of $C_{i+1}$, it follows that the
    output of $C_{i+1}$ can be simulated by a classical algorithm that makes at
    most
    \begin{equation}
        g(n) \cdot d
    \end{equation}
    queries. Therefore, from Lemma~\ref{lem:discovery-probability} it follows that
    \begin{equation}
        \tnorm{\ClassicalTierSimulator(C_{i+1}, x, \Vknown, T)\vert_\out
        - C_{i+1}(x)}
        \leq g(n) \cdot d \cdot \frac{2^{n+2}-2}{2^{2n}}.
    \end{equation}
\end{proof}

\begin{lemma}\label{lem:vknown-size}
    Let $(x, \Vknown)$ be the random variable produced by $\A^\zeta(T)$, for some
    $\zeta \in \{1,\dots,\eta\}$. Say $d$ is the depth of $C_{\zeta+1}$, $g(n)$ is
    the width of $C_{\zeta+1}$, and $\abs{\Vknown}$ is the number of vertices in
    $\Vknown$. Then it holds that
    \begin{equation}\label{eq:vknown-at-each-tier}
        \abs{\Vknown} \leq 4^{(d+1)\zeta} (g_0(n) \cdot d_0),
    \end{equation}
    where $g_0(n)$ and $d_0$ are the maximum width and depth of the classical
    tiers in $C^{\zeta + 1}$.
\end{lemma}
\begin{proof}
    We will prove this by induction on $\zeta \in \{1,\dots,\eta\}$. Let
    $\Vknown^\zeta$ be the value of $\Vknown$ produced by $\A^\zeta(T)$.

    \paragraph{Base Case.} Suppose $\zeta = 1$. By definition, the first tier is a
    classical tier, and initially $\abs{\Vknown} = 1$. By application of
    Lemma~\ref{lem:classical-composition}, we get
    Equation~\eqref{eq:vknown-at-each-tier}.

    \paragraph{Induction Case.} Suppose that
    Equation~\eqref{eq:vknown-at-each-tier} is true for $\zeta$, we will show that
    it also holds for $\zeta + 1$.

    \subparagraph{Case 1.} Suppose that $C_{\zeta + 1}$ is a classical
    tier. Combining Lemma~\ref{lem:classical-composition} with the induction
    hypothesis we get
    \begin{align}
        \bigabs{\Vknown^{\zeta+1}}
    &\leq \bigabs{\Vknown^{\zeta}} + (g(n) \cdot d)\\
    &\leq 4^{(d+1)\zeta} (g_0(n) \cdot d_0) + (g_0(n) \cdot d_0)\\
    &\leq (4^{(d+1)\zeta} + 1) (g_0(n) \cdot d_0)\\
    &\leq 4^{(d+1)(\zeta + 1)}(g_0(n) \cdot d_0),
    \end{align}
    as desired.

    \subparagraph{Case 2.} Suppose that $C_{\zeta + 1}$ is a quantum
    tier. Combining Lemma~\ref{lem:quantum-composition} with the induction
    hypothesis we get
    \begin{align}
        \bigabs{\Vknown^{\zeta+1}}
    &\leq \bigabs{\Vknown^{\zeta}} + 4^{d}\bigabs{\Vknown^{\zeta}}\\
    &\leq (1+4^d)\bigl(4^{(d+1)\zeta} (g_0(n) \cdot d_0)\bigr)\\
    &\leq 4^{d+1}\bigl(4^{(d+1)\zeta} (g_0(n) \cdot d_0)\bigr)\\
    &\leq 4^{(d+1)(\zeta + 1)}(g_0(n) \cdot d_0),
    \end{align}
    as desired.
\end{proof}

\begin{lemma}\label{lem:inductive-step}
    Given a $(n,\eta,4^d,d,g(n))$-hybrid-quantum circuit $C(T)$ with a hardcoded input
    (which is $0^n$). For all $\zeta \in \{1,\dots,\eta\}$, if
    \begin{equation}
        \tnorm{\A^\zeta(T)\vert_{\out} - C^\zeta(T)} \leq B
    \end{equation}
    for some bound $B$, then
    \begin{equation}
        \tnorm{\A^{\zeta+1}(T)\vert_{\out} - C^{\zeta+1}(T)}
        \leq (g(n))^2 \cdot 4^{(\zeta+2)(d+1)} \cdot \frac{2^{n+2}-2}{2^{2n}} + B.
    \end{equation}
\end{lemma}
\begin{proof}
    Define the random variable $x \coloneqq \A^\zeta(T)\vert_{\out}$.  By
    assumption, we have that
    \begin{equation}\label{eq:x-is-close-to-cit}
        \tnorm{x - C^\zeta(T)} = \tnorm{\A^\zeta(T)\vert_{\out} - C^\zeta(T)} \leq B
    \end{equation}

    Using the fact that, by definition, $C^{\zeta+1}(T) = C_{\zeta+1}(C^\zeta(T))$
    and applying the triangle inequality we get
    \begin{align}
        I \coloneqq \tnorm{\A^{\zeta+1}(T)\vert_{\out} - C^{\zeta+1}(T)}
    &= \tnorm{\A^{\zeta+1}(T)\vert_{\out} - C_{\zeta+1}(C^{\zeta}(T))}\\
    &\leq \tnorm{\A^{\zeta+1}(T)\vert_{\out} - C_{\zeta+1}(x)}
    + \tnorm{C_{\zeta+1}(x) - C_{\zeta+1}(C^{\zeta}(T))}.
    \end{align}
    Since we can interpret $C_{\zeta+1}$ as a quantum channel, and the trace norm is
    nonincreasing under quantum channels (see Theorem 9.2 in Nielsen and
    Chuang~\cite{NielsenC10}) we get the bound
    \begin{equation}
        I \leq \tnorm{\A^{\zeta+1}(T)\vert_{\out} - C_{\zeta+1}(x)} + \tnorm{x - C^{\zeta}(T)}.
    \end{equation}
    Plugging in Equation~\eqref{eq:x-is-close-to-cit} we get
    \begin{equation}
        I \leq \cdist{\A^{\zeta+1}(T)\vert_{\out}}{C_{i+1}(x)} + B.
    \end{equation}

    We now have two cases.

    \subparagraph{Case 1.} $C_{\zeta+1}$ is a classical tier. Then
    $\A^{\zeta+1}(T) = \ClassicalTierSimulator(C_{\zeta+1}, \A^\zeta(T), T)$, we get
    \begin{equation}
        I \leq
        \tnorm{\ClassicalTierSimulator(C_{\zeta+1},
        \A^\zeta(T), T)\vert_\out - C_{\zeta+1}(x)} + B.
    \end{equation}
    Applying Lemma \ref{lem:classical-composition} (recall that, by assumption,
    the maximum depth of a classical circuit is $4^d$), we get
    \begin{align}
        I
    &\leq g(n) \cdot 4^d\abs{\Vknown} \cdot \frac{2^{n+2}-2}{2^{2n}} + B,
    \end{align}
    using Lemma~\ref{lem:vknown-size} we obtain
    \begin{align}
        I
    &\leq g(n) \cdot 4^d \cdot \Bigl( 4^{\zeta (d+1)} g(n) d \Bigr)
    \cdot \frac{2^{n+2}-2}{2^{2n}} + B\\
    &\leq g(n) \cdot 4^d \cdot \Bigl( 4^{\zeta (d+1)} g(n) 4^d \Bigr)
    \cdot \frac{2^{n+2}-2}{2^{2n}} + B\\
    &\leq (g(n))^2 \cdot 4^{(\zeta+2)(d+1)} \cdot \frac{2^{n+2}-2}{2^{2n}} + B,
    \end{align}
    as desired.

    \subparagraph{Case 2.} $C_{\zeta+1}$ is a quantum tier. Then
    $\A^{\zeta+1}(T) = \QuantumTierSimulator(C_{\zeta+1}, \A^\zeta(T), T)$, we get
    \begin{equation}
        I \leq \tnorm{\QuantumTierSimulator(C_{\zeta+1}, \A^\zeta(T), T)\vert_\out
        - C_{\zeta+1}(x)} + B.
    \end{equation}
    Applying Lemma~\ref{lem:quantum-composition}, we get
    \begin{align}
        I
    &\leq 4 g(n) \cdot 4^{d+1} \abs{\Vknown} \cdot \frac{2^{n+2}-2}{2^{2n}} +
    B\\
    &\leq g(n) \cdot 4^{d+2} \abs{\Vknown} \cdot \frac{2^{n+2}-2}{2^{2n}} + B,
    \end{align}
    using Lemma~\ref{lem:vknown-size} we obtain
    \begin{align}
        I
    &\leq g(n) \cdot 4^{d+2} \cdot \Bigl( 4^{\zeta (d+1)} g(n) d \Bigr)
    \cdot \frac{2^{n+2}-2}{2^{2n}} + B\\
    &\leq g(n) \cdot 4^{d+2} \cdot \Bigl( 4^{\zeta (d+1)} g(n) 4^d \Bigr)
    \cdot \frac{2^{n+2}-2}{2^{2n}} + B\\
    &\leq (g(n))^2 \cdot 4^{(\zeta+2)(d+1)} \cdot \frac{2^{n+2}-2}{2^{2n}} + B,
    \end{align}
    as desired.
\end{proof}

\begin{lemma} \label{lem:siman} Consider a $(n,\eta,4^d,d)$-hybrid-quantum
    circuit $C(T)$. For all $\zeta \in \{1,\dots,\eta\}$, it holds that
    \begin{equation}\label{eq:sim-is-close-to-circuit}
        \cdist{\A^\zeta(T)\vert_{\out}}{C^\zeta(T)}
        \leq 4^{(\zeta+2)(d+2)} \cdot (g(n))^2 \cdot \frac{2^{n+2}-2}{2^{2n}},
    \end{equation}
    where $g(n)$ is the width of $C^\zeta(T)$.
\end{lemma}

\begin{proof}
    The proof proceeds by induction.  We will begin with the base case $\zeta = 0$,
    and the statement for each larger $\zeta$ will be proven assuming the statement
    for $\zeta - 1$.

    \paragraph{Base Case.} Suppose $\zeta = 0$, then it is immediate that
    Equation~\eqref{eq:sim-is-close-to-circuit} holds.

    \paragraph{Inductive Case.} Suppose
    Equation~\eqref{eq:sim-is-close-to-circuit} holds for some $\zeta$, then
    Lemma~\ref{lem:inductive-step} implies that
    \begin{align}
        \cdist{\A^{\zeta+1}(T)\vert_{\out}}{C^{\zeta+1}(T)}
    &\leq (g(n))^2 \cdot 4^{((\zeta+1)+2)(d+1)} \cdot \frac{2^{n+2}-2}{2^{2n}} +
    (g(n))^2 \cdot 4^{(\zeta+2)(d+1)} \cdot \frac{2^{n+2}-2}{2^{2n}}\\
    &\leq \Bigl(4^{((\zeta+1)+2)(d+1)} + 4^{(\zeta+2)(d+1)}\Bigr)
    (g(n))^2 \cdot \frac{2^{n+2}-2}{2^{2n}}\\
    &\leq \Bigl(4^{((\zeta+1)+2)(d+1)} + 4^{(\zeta+2)(d+1)}\Bigr)
    (g(n))^2 \cdot \frac{2^{n+2}-2}{2^{2n}}\\
    &\leq 4^{((\zeta+1)+2)(d+2)} \cdot (g(n))^2 \cdot \frac{2^{n+2}-2}{2^{2n}},
    \end{align}
    so Equation~\eqref{eq:sim-is-close-to-circuit} also holds for $\zeta+1$. This
    completes the inductive step.
\end{proof}

We now have the tools to prove Theorem~\ref{thm:classical-sim}.

\begin{proof}[Proof of Theorem~\ref{thm:classical-sim}]
    Use Lemma~\ref{lem:siman} with $\zeta = \eta$ to get the bound on the error.
    Combining the observation that we can replace $\A(T)$ by an algorithm that
    makes at most $\abs{\Vknown}$ queries with Lemma~\ref{lem:vknown-size} we
    get the bound on the number of queries.
\end{proof}

\section{Analysis of the Jozsa Simulator}
\label{sec:analysis-of-jozsa-simulator}

\begin{lemma}
    \label{lem:jozsa-composition}
    Suppose $C$ is a Jozsa quantum circuit. Let $(\ket{\psi}, \Vknown)$ be the
    output produced by $\A(T)$. Say $g(n)$ is the width of $C$ and $c(n)$ the
    classical depth (the maximum depth of a classical circuit embedded) of $C$,
    and $\abs{\Vknown}$ is the number of vertices in $\Vknown$. Further, assume
    that
    \begin{equation}\label{eq:jozsa-composition-assumption}
        4^d (\abs{\Vknown} + d c(n) g(n)) = \subexp(n).
    \end{equation}
    Then the following statements hold.
    \begin{enumerate}
        \item It holds that\label{item:j1}
            \begin{align}
                \tnorm{
                    \JozsaQuantumTierSimulator(C, \ket{\psi}, \Vknown, T)\vert_\out
                - C(x) }\\
                \leq 4 c(n) g(n) \cdot 4^{2d+1} \abs{\Vknown} \cdot
                \frac{2^{n+2}-2}{2^{2n}}.
            \end{align}
        \item $\JozsaQuantumTierSimulator(C, \ket{\psi}, \Vknown, T)$ makes
            at most $4^d (\abs{\Vknown} + d c(n) g(n))$ queries to
            $T$.\label{item:j2}
    \end{enumerate}
\end{lemma}
\begin{proof}
    Let us denote by $\C[\ell]$ the intermediate quantum state (on registers
    \texttt{R1} and \texttt{R2}) produced at the end of the first $\ell$ layers.
    Let $\ket{\psi_i}$ and $\ket*{\Vknown^i}$ be defined as in
    Algorithm~\ref{alg:JozsaQuantumTierSimulator}.

    We will show, by induction, that at every stage $i \leq 2d$ of the Jozsa
    circuit, it holds that
    \begin{equation}\label{eq:jozsa-bound-on-vknown}
        \abs*{\Vknown^i} \leq 4^i(\abs{\Vknown} + c(n)g(n))
    \end{equation}
    and
    \begin{equation}\label{eq:jozsa-bound-on-tnorm}
        \tnorm{\ket{\psi_i} - C[i]} \leq 4 c(n) g(n)
        \left( \sum_{i=0}^{i} \abs{\Vknown^i} \right)
        \cdot \frac{2^{n+2}-2}{2^{2n}}.
    \end{equation}
    Combining \eqref{eq:jozsa-bound-on-vknown} and
    \eqref{eq:jozsa-bound-on-tnorm} and using sum of geometric series we get
    part~\ref{item:j1} of the lemma. Part~\ref{item:j2} follows from
    Equation~\eqref{eq:jozsa-bound-on-vknown} and the fact that
    $\JozsaQuantumTierSimulator$ can be simulated by an algorithm that makes at
    most $\abs*{\Vknown}$ queries.

    \paragraph{Base Case.} Suppose $i=0$, then \eqref{eq:jozsa-bound-on-vknown}
    and \eqref{eq:jozsa-bound-on-tnorm} are immediate.

    \paragraph{Induction Case.} Suppose that for a given $i \geq 0$
    \eqref{eq:jozsa-bound-on-vknown} and \eqref{eq:jozsa-bound-on-tnorm} hold.
    We will now prove that they also hold for $i+1$.

    From the triangle inequality we get
    \begin{align}
        I \coloneqq \tnorm{\ket{\psi_{i+1}} - C[i+1]}
    &\leq \tnorm{\ket{\psi_{i+1}} - L_{i+1}\ket{\psi_i}} +
    \tnorm{L_{i+1}\ket{\psi_i} - C[i+1]}\\
    &= \tnorm{\ket{\psi_{i+1}} - L_{i+1}\ket{\psi_i}} +
    \tnorm{L_{i+1}\ket{\psi_i} - L_{i+1}C[i]}.
    \end{align}
    Interpreting $L_{i+1}$ as a quantum channel, and utilizing the fact that the
    trace norm is nonincreasing under quantum channels (see Theorem 9.2 in Nielsen
    and Chuang~\cite{NielsenC10}) we get
    \begin{equation}
        I
        \leq \tnorm{\ket{\psi_{i+1}} - L_{i+1}\ket{\psi_i}} +
        \tnorm{\ket{\psi_i} - C[i]}
    \end{equation}
    Applying the induction hypothesis (Equation~\eqref{eq:jozsa-bound-on-tnorm}),
    we get
    \begin{equation}\label{eq:last-common-step}
        I
        \leq \tnorm{\ket{\psi_{i+1}} - L_{i+1}\ket{\psi_i}} +
        4 c(n) g(n) \left( \sum_{i=0}^{i} \abs{\Vknown^i} \right)
        \cdot \frac{2^{n+2}-2}{2^{2n}}.
    \end{equation}
    From here, we have two cases.

    \subparagraph{Case 1: $L_i$ is a quantum layer.}
    Applying Lemma~\ref{lem:quantumlayersimulator-works} to
    Equation~\eqref{eq:last-common-step}, we get
    \begin{align}
        I
    &\leq 4 g(n)\abs{\Vknown^{i+1}} \cdot \frac{2^{n+2}-2}{2^{2n}}
    + 4 c(n) g(n) \left( \sum_{i=0}^{i} \abs{\Vknown^i} \right)
    \cdot \frac{2^{n+2}-2}{2^{2n}}\\
    &\leq 4 c(n) g(n) \left( \sum_{i=0}^{i+1} \abs{\Vknown^i} \right)
    \cdot \frac{2^{n+2}-2}{2^{2n}}.
    \end{align}
    This proves \eqref{eq:jozsa-bound-on-tnorm} for $i+1$.

    Using Lemma~\ref{lem:quantumlayersimulator-works} and the induction
    hypothesis (Equation~\eqref{eq:jozsa-bound-on-vknown}) we get
    \begin{align}
        \abs*{\Vknown^{i+1}}
    &\leq 4 \abs*{\Vknown^i}\\
    &\leq 4 \cdot 4^i(\abs{\Vknown} + c(n)g(n))\\
    &\leq 4^{i+1}(\abs{\Vknown} + c(n)g(n))
    \end{align}
    This proves \eqref{eq:jozsa-bound-on-vknown} for $i+1$.

    \subparagraph{Case 1: $L_i$ is a classical tier.}
    Applying Lemma~\ref{lem:classical-composition} to
    Equation~\eqref{eq:last-common-step}, we get
    \begin{align}
        I
    &\leq g(n) c(n) \abs{\Vknown^{i+1}} \cdot \frac{2^{n+2}-2}{2^{2n}}
    + 4 c(n) g(n) \left( \sum_{i=0}^{i} \abs{\Vknown^i} \right)
    \cdot \frac{2^{n+2}-2}{2^{2n}}\\
    &\leq 4 c(n) g(n) \left( \sum_{i=0}^{i+1} \abs{\Vknown^i} \right)
    \cdot \frac{2^{n+2}-2}{2^{2n}}.
    \end{align}
    This proves \eqref{eq:jozsa-bound-on-tnorm} for $i+1$.

    Using Lemma~\ref{lem:quantumlayersimulator-works} and the induction
    hypothesis (Equation~\eqref{eq:jozsa-bound-on-vknown}) we get
    \begin{align}
        \abs*{\Vknown^{i+1}}
    &\leq g(n) c(n) + \abs*{\Vknown^i}\\
    &\leq g(n) c(n) + 4^i(\abs{\Vknown} + c(n)g(n))\\
    &\leq 4^{i+1}(\abs{\Vknown} + c(n)g(n))
    \end{align}
    This proves \eqref{eq:jozsa-bound-on-vknown} for $i+1$.
\end{proof}

We now have the tools to prove Theorem~\ref{thm:jozsa-classical-sim}.

\begin{proof}[Proof of Theorem~\ref{thm:jozsa-classical-sim}]
    Use Lemma~\ref{lem:jozsa-composition} to get the bound on the error and a
    bound on the number of queries.
\end{proof}

\section{Analysis of the Bottleneck Simulator}
\label{sec:analysis-of-bottleneck-simulator}

The following two lemmas are immediate from the definition of $\T_{V,x,r}^i$ and
$\T_V$ (Definition~\ref{def:TVsr}.)

\begin{lemma}\label{lem:rough-lower-bound}
    When we add a label $b$ to $\Vknown$ to get the new dictionary
    $\Vnknown$, it holds that
    \begin{equation}
        \Bigabs{\T_{\Vnknown, x,r_{\leq i}}^i}
        \geq 
        \Pr_{P \in \T_{\Vknown, x,r_{\leq i}}^i}[\text{P is consistent with $b$}]
        \cdot
        \Bigabs{\T_{\Vknown, x, r_{\leq i}}^i}.
    \end{equation}
\end{lemma}

\begin{lemma}\label{lem:rough-upper-bound}
    For any $V$, $x$, $r$, $i$, it holds that
    \begin{equation}
        \abs{\T_{V, x,r_{\leq i}}^i}
        \leq \abs{\T_{V}}.
    \end{equation}
\end{lemma}

\begin{lemma}\label{lem:tv-char}
    For any dictionary $V$, it holds that
    \begin{equation}
        \abs{\T_{V}} = \frac{(2^{2n} - \abs{V})!}{(2^n - \abs{V})!}.
    \end{equation}
\end{lemma}
\begin{proof}
    By definition $\T_{V}$ contains all the possible Welded Tree blackboxes that
    are consistent with the vertex labellings in dictionary $V$. Recall that,
    given the labellings in $V$, there are still $2^{2n} - \abs{V}$ unused
    labels and $2^n - \abs{V}$ vertices that need a label. In other words,
    $\abs{\T_{V}}$ is equal to the number of $(2^n - \abs{V})$-permutations of
    $2^{2n} - \abs{V}$, which is
    \begin{equation}
        \frac{(2^{2n} - \abs{V})!}{(2^n - \abs{V})!}.
    \end{equation}
\end{proof}

\begin{lemma}\label{lem:less-rough-upper-bound}
    When we add a label $b$ to
    $\Vknown$ to get $\Vnknown$, we get
    \begin{equation}
        \abs{\T_{\Vknown^{\textup{new}}, x,r}^i}
        \leq \frac{1}{2^{2n}-\abs{\Vknown}}\abs{\T_{\Vknown}}.
    \end{equation}  
\end{lemma}
\begin{proof}
    It follows from Lemma \ref{lem:tv-char} that
    \begin{equation}
    \abs{\T_{\Vknown^{\textup{new}}}}\leq \frac{1}{2^{2n}-\abs{\Vknown}}\abs{\T_{\Vknown}}
    \end{equation}
    The proof follows by application of Lemma~\ref{lem:rough-upper-bound} to
    this.
\end{proof}

\begin{lemma}\label{lem:keylemma1}
    In the context of algorithm $\M^k$, when the subroutine $\bottleneck(i, x,
    \Vcknown, \Vhknown)$ does not $\abort$,   the set $\Vknown := \bottleneck(i,
    x, \Viknown, \Vfknown)$ has size at most $\abs*{\Vcknown} + 2n(g(n)+
    \abs{r})$, and has the property that,

    \begin{align}
        &\forall b \in \{0,1\}^{2n}
        \text{ such that b does not appear in  }\Vknown:\label{eq:badguessit1}\\
        &\mathbb{P}_{P \in \T_{\Vknown, x,r_{\leq i}}^i} 
        \left[b \text{ is a valid vertex in P} \right] \leq 2^{-n/100}\label{eq:badguessit2}
    \end{align} 
\end{lemma}
\begin{proof}
    The \while{} loop within $\bottleneck(i, x, \Vcknown, \Vhknown)$ is defined
    to continue iterating until Equations \eqref{eq:badguessit1},
    \eqref{eq:badguessit2} are true.  So, under our assumption that the
    \while{} loop returns a valid $\Vknown$, without ever calling the $\abort$
    command, we know that Equations \eqref{eq:badguessit1}
    \eqref{eq:badguessit2} are satisfied.  Since $\Vknown$ grows in size by at
    most $1$ label for each iteration of the \while{} loop, it suffices to prove
    that, so long as no $\abort$ occurs, the \while{} loop finishes and returns
    a valid answer $\Vknown$ after at most $2(g(n)+ \abs{r})$ iterations. To
    prove this let $\Vknown^j$ denote the set $\Vknown$ as it is defined within
    the $j^{th}$ iteration of the \while{} loop.  We will show that the number
    of iterations $j$ cannot exceed $2(g(n)+ \abs{r})$ by showing both a
    necessary upper bound and a necessary lower bound on the quantity
    $\Bigabs{\T_{\Vknown^j, x,r_{\leq i}}^i}$, which contradict each other when $j$
    exceeds $2(g(n)+ \abs{r})$.

    \paragraph{The Lower Bound}
    By repeated application of Lemma~\ref{lem:rough-lower-bound} we get
    \begin{equation}\label{eq:lowerbound1}
        \Bigabs{\T_{\Vknown^{j}, x, r_{\leq i}}^i}
        \geq 2^{-\frac{nj}{100}}\Bigabs{\T_{\Vcknown, x, r_{\leq i}}^i}.
    \end{equation}
    Since we are assuming that $\bottleneck(i, x, \Vcknown, \Vhknown)$ did not
    $\abort$, we know from line \ref{algline:haltifsmall1} and
    \ref{algline:haltifsmall2} of $\bottleneck$ that it must be the case that
    \begin{equation}
        \Bigabs{\T_{\Vcknown,x,r_{\leq i}}^i}
        \geq 2^{-n(g(n)+\abs{r})} \Bigabs{\T_{\Vcknown}}.
    \end{equation}
    Combining this with equation \eqref{eq:lowerbound1}
    gives,
    \begin{equation}\label{eq:Tvxr-lower-bound}
        \Bigabs{\T_{\Vknown^{j}, x, r_{\leq i}}^i}
        \geq 2^{-\frac{nj}{100} -n(g(n)+\abs{r})}\abs{\T_{\Vcknown}}.
    \end{equation}

    \paragraph{The Upper Bound}
    By repeated application of
    Lemma~\ref{lem:less-rough-upper-bound} we get
    \begin{equation}
        \abs{\T_{\Vknown^{j}, x, r_{\leq i}}^i}
        \leq \left(\frac{1}{2^{2n}-\abs{\Vknown}}\right)^{j}\abs{\T_{\Vcknown}}.
    \end{equation}
    Since $\bigabs{\Vcknown} \leq 2^n$, we have
    \begin{align}
        \Bigabs{\T_{\Vknown^{j}, x, r_{\leq i}}^i}
        &\leq \left(\frac{1}{2^{2n}-2^n}\right)^{j}\abs{\T_{\Vcknown}}\\
        &\leq \left(\frac{1}{2^{n}(2^n-1)}\right)^{j}\abs{\T_{\Vcknown}}\\
        &\leq \left(\frac{1}{2^{n}}\right)^{j}\abs{\T_{\Vcknown}}\\
        &= (2^{-n})^{j}\abs{\T_{\Vcknown}}\\
        &= 2^{-nj}\abs{\T_{\Vcknown}}.\label{eq:Tvxr-upper-bound}
    \end{align}

    But notice that \eqref{eq:Tvxr-upper-bound} contradicts
    \eqref{eq:Tvxr-lower-bound} when say, $j \geq 2(g(n)+\abs{r})$. Therefore,
    we have shown that the \while{} loop cannot run for more than $k \leq
    2(g(n)+\abs{r})$ iterations, and the set $\Vknown^k$ at the end of the
    \while{} loop satisfies $\abs*{\Vknown^k} \leq \abs*{\Vcknown} +
    2(g(n)+\abs{r})$.

    The final step of the $\bottleneck$ subroutine connects every new vertex in
    $\Vknown^k$ to the $\entrance$ node to create $\Vknown^{\text{complete}}$.
    This requires an overhead of at most $n$, and so
    \begin{equation}
        \bigabs{\Vknown^{\text{complete}}} \leq \bigabs{\Vcknown} + 2n(g(n)+\abs{r}).
    \end{equation}
\end{proof}

\begin{lemma}\label{lem:queryupperbound}
    In the classical simulation algorithm $\M^k$ which simulates the first $k$
    tiers of $C(T)$, the set of all encountered vertices after $k$ tiers
    $\Vhknown = \M^k(T)[3]$ has size 
    \begin{equation}
        \abs*{\Vhknown} \leq k q(n) 2^{q(n)} 2n (g(n) + \abs{r})
    \end{equation}
\end{lemma}
\begin{proof}
    Lemma \ref{lem:keylemma1} shows that, within each call of
    $\BottleneckQuantumTierSimulator$, $\abs*{\Vknown^{0}} \leq 2n(g(n) +
    \abs{r})$.  Since $\SimulateOracle$ at most doubles the size of the set $\Vknown$
    that it acts on, and since, by Lemma \ref{lem:keylemma1}, $\bottleneck$ adds
    at most $2n(g(n) + \abs{r})$ to the size of the set $\Vknown$ that it acts
    on, we have that, within $\BottleneckQuantumTierSimulator$,
    $\abs*{\Vknown^{t}} \leq 2 \abs*{\Vknown^{t-1}} + 2n(g(n) + \abs{r})$.
    Following this recursion through $q(n)$ quantum layers yields
    \begin{equation}
        \abs{\Vknown^{q(n)}} \leq q(n) 2^{q(n)}2n(g(n) + \abs{r}).
    \end{equation}
    At the end of
    $\BottleneckQuantumTierSimulator$, $\Vknown^{q(n)}$ is merged with
    $\Vhknown$, thus increasing $\Vhknown$ in size by at most 
    \begin{equation}
        q(n) 2^{q(n)}2n(g(n) + \abs{r}).
    \end{equation}
    The subroutine
    $\BottleneckQuantumTierSimulator$ is run $k$ times in $\M^k$, and so, by the
    end of $\M^k$, we have that
    \begin{equation}
        \abs{\Vhknown} \leq k q(n) 2^{q(n)}2n(g(n) + \abs{r}),
    \end{equation}
    as desired.
\end{proof}

\begin{lemma}\label{lem:abort-probability}
    In the context of algorithm $\M^i$, the subroutine $\bottleneck(i, x,
    \Vcknown, \Vhknown)$ does not $\abort$ with probability higher than
    $2^{-n/8}$.
\end{lemma}
\begin{proof}	
    We bound the probability that $\bottleneck(i, x, \Viknown, \Vfknown)$ calls
    $\abort$ by considering the two separate $\abort$ cases.  

    \paragraph{Case 1. \abort{} in line~\ref{algline:haltifsmall2}}
    Let
    \begin{align}
        N \coloneqq \left\{ 
            x,r:
            \bigabs{\T_{\Vcknown,x,r_{\leq i}}^i} 
            < 2^{-n(g(n)+\abs{r})} \bigabs{\T_{\Vcknown}} 
        \right\}
    \end{align}
    denote the set of bitstring, seed tuples which are unlikely. In other words,
    $x,r \in N$ implies that, under seed $r$, the set of black-box trees
    consistent with $\M^i$ outputting $x$ and $\Vcknown$ is exponentially
    smaller than the set of all black-box trees consistent with the $\Vcknown$.

    By definition, probability of an $\abort$ in line \ref{algline:haltifsmall2}
    of $\bottleneck$ is the probability of receiving an unlikely bitstring.
    \begin{align}
        \mathbb{P}_{x,r|\Vcknown} \left [ \Bigabs{\T_{\Vcknown,x,r_{\leq i}}^i} <
        2^{-n(g(n)+\abs{r})} \abs{\T_{\Vcknown}} \right ] 
    &= \sum_{x,r \in N} \mathbb{P}\Bigl[x,r \,\Big\vert\,
    \Vcknown\Bigr]\\
    & = \sum_{x,r \in N} \frac{\Bigabs{\T_{\Vcknown,x,r}^i}}{\abs{\T_{\Vcknown}}}\\
    &\leq \sum_{x,r \in N} 2^{-n(g(n)+\abs{r})} \\
    & \leq \sum_{x,r} 2^{-n(g(n)+\abs{r})}\\
    &\leq 2^{\abs{x} + \abs{r}}2^{-n(g(n)+\abs{r})}\\
    &\leq 2^{g(n) + \abs{r}}2^{-n(g(n)+\abs{r})}\\
    &\leq 2^{-(n-1)(g(n)+\abs{r})}\\
    &\leq 2^{-(n-1)}
    \end{align}

    \paragraph{Case 2. \abort{} in line~\ref{algline:guessabort2}}
    In this case we must have a label $b'$ which is does not appear in $\Vknown
    \subseteq \Vhknown$, or $\Vhknown$, and yet satisfy
    \begin{equation}
        \mathbb{P}_{P \in \T_{\Vknown, x, r_{\leq i}}^i}
        \left[b' \text{ is a valid label in P} \right]
        \geq 2^{-n/100}.
    \end{equation}
    Since the entire algorithm up to this point has only made $\abs*{\Vhknown}$
    classical queries, we know from Lemma~\ref{lem:discovery-probability}, that
    with our fixed $r$ but over the randomness in the black-box, the probability
    of guessing a valid label outside of $\Vhknown$ is at most
    \begin{align}
        \abs*{\Vhknown} \cdot \frac{2^{n+2}-2}{2^{2n}} \leq  2^{-n/2}.
    \end{align}
    Here the inequality follows for sufficiently large $n$ because, by Lemma
    \ref{lem:queryupperbound} we have that $\abs*{\Vhknown} \leq k q(n)
    2^{q(n)}2n(g(n) + \abs{r})$ which is less than $ 2^{-n/2}$ for sufficiently
    large $n$ because $q(n) = \polylog(n)$, and $\abs{r} = \poly(n)$ (recall
    that $\abs{r} \leq nkq(n)g(n)$).  However, since (for our fixed $r$)
    $\Vknown$ and $x$ are deterministic functions of $\Vhknown$, it follows that
    the probability of arriving in the above situation can be at most
    $\frac{2^{-n/2}}{2^{-n/100}} \leq 2^{-n/4}$.  Any higher probability of this
    event would yield an algorithm for guessing a valid label outside of
    $\Vhknown$ that contradicts Lemma~\ref{lem:discovery-probability}.

    Note that lines \ref{algline:guessabort1} and \ref{algline:guessabort2}
    are run at most $2(g(n) + \abs{r})$ times.  Thus, union bounding over the
    two different $\abort$ cases gives an upper bound on the total abort
    probability of $2(g(n) + \abs{r})2^{-n/4} + 2^{-(n-1)} \leq 2^{-n/8}$ where
    the inequality holds for sufficiently large $n$ because $g(n)$ and $\abs{r}$
    scale polynomially.
\end{proof}

\begin{lemma}\label{lem:keylemma2}
    Consider the states $\ket{\phi_\ell}$ and $\ket{\psi_\ell}$ in the
    subroutine $\BottleneckQuantumTierSimulator(C_i, i-1,y, \Viknown,
    \Vknown^{hist}, T$) in $\M^i$.  As long as the subroutine $\bottleneck$ does
    not $\abort$ in the entire course of $\M^i$ up to this point, then these two
    states satisfy:
    \begin{align}
        \|\ket{\psi_\ell} - L_\ell^T \ket{\phi_\ell} \|_1  \leq 2^{-n/100}
    \end{align}
\end{lemma}
\begin{proof}
    It follows from Lemma \ref{lem:keylemma1} that, so long as $\bottleneck$
    does not $\abort$, no algorithm whose knowledge of the blackbox $T$ is
    limited to 
    \begin{equation}
        \Vknown^{\ell} \coloneqq \bottleneck(j-1, x, \Vknown^{hist},
        \Vknown^{\ell\text{-temp}})
    \end{equation}
    and the string $x = \M^{i-1}(C(T), T)[1]$, can guess a valid label outside
    of $\Vknown^{\ell}$ with probability greater than $2^{-n/100}$.  There is
    simply too great a variety of blackboxes $P$ consistent with
    $\Vknown^{\ell}$ and $x = \M^{i-1}(C(P), P)[1]$ for such a guessing scheme
    to be possible.  

    Consider the algorithm $\mathcal{B}$ described in the proof of Lemma
    \ref{lem:quantumlayersimulator-works}.  Let $\mathcal{B}$ use the classical
    description of $\ket{\phi_\ell} \equiv \sum_z c_z \ket{z}$, the layer
    $L_\ell^T$, and the sets  $\Vknown^{\ell-1}$, $\Vknown^{\ell\text{-temp}}$
    in place of $\ket{\phi}$,  $L^T$,  $\Vknown$, $\Vknown'$ (respectively) in
    its original definition.  As shown in Lemma
    \ref{lem:quantumlayersimulator-works}, this algorithm guesses a valid label
    outside of $\Vknown^{\ell}$ with probability  at least 
    \begin{equation}
    \frac{1}{g(n)}
    \sum_{z \in \Outliers} |c_z|^2 
    = \fid(\ket{\psi_\ell}, L_\ell^T \ket{\phi_\ell}) 
    \geq \sqrt{\frac{1}{2}\|\ket{\psi_\ell} - L_\ell^T \ket{\phi_\ell} \|_1}.
    \end{equation}
    Furthermore, we will show below that this algorithm's knowledge of the
    blackbox $T$ is limited to $\Vknown^{\ell} \coloneqq \bottleneck(j-1, x,
    \Vknown^{hist}, \Vknown^{\ell\text{-temp}})$ and the string $x =
    \M^{i-1}(C(T), T)[1]$.    It follows by Lemma \ref{lem:keylemma1} that
    $\|\ket{\psi_\ell} - L_\ell^T \ket{\phi_\ell} \|_1  \leq 2^{-n/50+1} \leq
    2^{-n/100}$ (for sufficiently large $n$), which is the desired result.  

    To complete the proof: The reason that $\mathcal{B}$ is only a function of
    $\Vknown^{\ell} \coloneqq \bottleneck(j-1, x, \Vknown^{hist},
    \Vknown^{\ell\text{-temp}})$ and the string $x = \M^{i-1}(C(T), T)[1]$ is
    that $\mathcal{B}$ only takes input ($\ket{\phi_\ell}$,  $L_\ell^T$,
    $\Vknown^{\ell-1}$, $\Vknown^{\ell\text{-temp}}$) and each of these can be
    constructed from $\Vknown^{\ell}$ and $x$.  To see this, note that
    $L_\ell^T$ is just the circuit diagram of a layer, it is not a function of
    the blackbox $T$ at all, it simply illustrates where queries to $T$ are made
    in the circuit.  Secondly,  $\Vknown^{\ell-1}$, and
    $\Vknown^{\ell\text{-temp}}$ are subsets of $\Vknown^{\ell}$ by definition
    (see the pseudocode of $\bottleneck$ and $\BottleneckQuantumTierSimulator$).
    Finally, the classical description of $\ket{\phi_\ell}$ is computed by
    starting with $x$ and proceeding through $\ell$ iterations of the loop in
    $\BottleneckQuantumTierSimulator$.  Since, within that loop, $\Vknown^{k}
    \subset \Vknown^{\ell}$ for $k \leq \ell$, the desired results follows.
\end{proof}

\begin{theorem}\label{thm:maintechthm}
    Let $\T$ be the set of all Welded Tree black-boxes.  Given a
    \emph{$(n,\ell(n),q(n), g(n))$-hybrid-quantum circuit} $C(T)$, which solves
    the Welded Tree problem with probability $p$ we have that 
    \begin{equation}
        \mathbb{P}_{T \in \T}[C(\M^{\ell(n)}(C(T), T)) 
        \text{ solves the Welded Tree problem}] 
        \geq p - 2^{-n/400}
    \end{equation}
\end{theorem}
\begin{proof}
    Recall the analysis of algorithm $\A^i$.  That same analysis could be used
    to analyze $\M^i$ except for one point, at Equation \eqref{eq:boundbyguess}
    in Lemma \ref{lem:quantumlayersimulator-works}, where we use that the set
    $\Vknown^{\ell+1}$ has the property that the probability of guessing a valid
    label outside of $\Vknown^{\ell+1}$ is at most $\abs*{\Vknown^{\ell+1}}
    \cdot \frac{2^{n+2}-2}{2^{2n}}$.  In the analysis of $\M^i$, within the
    subroutine $\BottleneckQuantumTierSimulator$ the set $\Vknown^{\ell+1}$ is
    defined differently than in $\A^i$ (it is modified by the subroutine
    $\bottleneck$), and so Equation \eqref{eq:boundbyguess} no longer holds.  To
    fix this, we replace Lemma \ref{lem:quantumlayersimulator-works} with
    Lemma~\ref{lem:keylemma2}, and continue the rest of the proof as originally
    specified in Section \ref{sec:few-tiers}. However, it follows from Lemma
    \ref{lem:keylemma1} that, so long as $\bottleneck$ does not $\abort$, the
    appropriate analog of the key Equation \eqref{eq:updated-tnorm} in the
    statement of Lemma \ref{lem:quantumlayersimulator-works} still holds in our
    new context.  Therefore, in the case that $\bottleneck$ never calls $\abort$
    in the entire course of $\M^i$, we may use the same analysis for $\M^i$ as
    we did for $\A^i$, except replacing Equation~\eqref{eq:updated-tnorm} with
    the statement of Lemma~\ref{lem:keylemma2}.  It follows that, in the absence
    of an $\abort$, an error in the simulation of at most $2^{-n/200}$ (in the
    trace norm) is incurred for every quantum layer.  Since there are $k q(n)$
    quantum layers in $C^k(T)$, this analysis gives
    \begin{equation}
        \cdist{\M^{k}(T)[1]}{C^k(T)}
        \leq k q(n) 2^{-n/200} + \mathbb{P}[\M^{k} \text{ calls } \abort],
    \end{equation}

    We know from Lemma~\ref{lem:abort-probability} that $\bottleneck$ has less than
    $2^{-n/8}$ probability of $\abort$ every time that it is called in $\M^k$.
    Since $\bottleneck$ is only called once per layer, and  $k q(n)$ quantum
    layers in $C^k(T)$ it follows by union bound that
    \begin{equation}
        \mathbb{P}[\M^{k} \text{ calls } \abort] \leq k q(n) 2^{-n/8}
    \end{equation}
    Thus,
    \begin{align}
        \cdist{\M^{k}(T)[1]}{C^k(T)}
 &\leq k q(n) 2^{-n/200}+\mathbb{P}[\M^{k} \text{ calls } \abort]\\
 &\leq  k q(n) 2^{-n/200}+k q(n) 2^{-n/8}\\
 &\leq 2 k q(n) 2^{-n/200}\\
 &\leq \ell(n) q(n) 2^{-n/200},
    \end{align}
    and the desired result follows for sufficiently large $n$.
\end{proof}

We now have the tools to prove Theorem~\ref{thm:mainthm-formal}.

\begin{proof}[Proof of Theorem~\ref{thm:mainthm-formal}]
    Suppose that
    \begin{equation}
        \H(T) = \{H_n : n \in \NN \},
    \end{equation}
    is a $\HQC$ algorithm. Then, by definition, $\H(T)$ is a polynomial-time
    uniform family of $(n, \eta, q(n), g(n))$-hybrid-quantum circuits (as
    defined in Definition~\ref{def:wlog-hybrid-quantum-circuit}) for some
    polynomials $\eta, g(n)$ and polylogarithm $q(n)$, querying $T$.

    Note that, by Lemma \ref{lem:queryupperbound} we know that the size of
    $\Vhknown$ at the end of $\M^k$ is at most $k q(n) 2^{q(n)}2n(g(n) +
    \abs{r})$.  This quantity is pseudopolynomial because $q(n) = \polylog(n)$,
    and $\abs{r} = \poly(n)$ (recall that $\abs{r} \leq n\eta q(n)g(n)$).  With
    this observation, the desired result follows for sufficiently large $n$ by
    combining Theorem \ref{thm:maintechthm} with the classical lower bound
    Theorem \ref{thm:class-lower-bound}.
\end{proof}

\end{document}